%% file: paper.tex
\documentclass[sigplan,10pt]{acmart}

\setcopyright{none}
\renewcommand\footnotetextcopyrightpermission[1]{}

\newcommand{\KEYWORDS}{Keywords}
\newcommand{\CONFERENCE}{Conference}
\newcommand{\PAGENUMBERS}{yes} 
\newcommand{\COMMENTS}{no}

\newcommand{\showComments}{yes}

\newcommand{\system}{DynamiQ\xspace}
\definecolor{steel_blue}{RGB}{70, 130, 180}
\usepackage{amsmath}
\usepackage{pifont}
\usepackage[linesnumbered, ruled, vlined]{algorithm2e}
\usepackage{fancybox}
\usepackage{pseudocode}
\usepackage{algpseudocode}
\usepackage{booktabs} 
\usepackage{outlines}
\usepackage{listings} 
\usepackage{tabularx}
\usepackage{colortbl, graphicx}
\usepackage{subcaption}
\usepackage{balance}
\usepackage{paralist}

\DeclareMathOperator*{\argmin}{arg\,min}

\usepackage{pifont}
\newcommand{\cmark}{\ding{51}}%
\newcommand{\xmark}{\ding{55}}%

\definecolor{darkgreen}{rgb}{0.0, 0.65, 0.0}

\ifthenelse{\equal{\COMMENTS}{no}}{%
\newcommand{\ag}[1]{\textit{\textcolor{blue}{[arpit]: #1}}} 
\newcommand{\nf}[1]{\textcolor{red}{(\textit{\textcolor{red}{NF: #1}} )}}
\newcommand{\rjh}[1]{\textit{\textcolor{darkgreen}{[rob]: #1}}} 
\newcommand{\rb}[1]{\textit{\textcolor{cyan}{[rohan]: #1}}} 
}
{
\newcommand{\ag}[1]{}
\newcommand{\lv}[1]{}
\newcommand{\rmnote}[1]{}
\newcommand{\rbnote}[1]{}
\newcommand{\nf}[1]{}
\newcommand{\mcnote}[1]{}
\newcommand{\rjh}[1]{}
\newcommand{\rb}[1]{}
}

\usepackage{xspace}
\newcommand*{\eg}{e.g.,\@\xspace}

\newcommand*{\ie}{\textit{i.e.,}\@\xspace}

\def \compress {}

\lstdefinelanguage{Python}
{keywords={MapBolt, ReduceBolt, >>, >+, +, &, push, if_, elif_, else, pop, 
match, fwd, modify, set, mod, sample, sampleD, sampleS, drop, $\triangleleft$, announce, 
withdraw, mapD, mapS, map, reduce, filter, filterD, filterS, runningReduceD, runningReduceS, 
distinct, toList, mapValues, countS, countByWindow, countByValueAndWindow, reduceByKey, window,
transform}, 
  sensitive=true, alsoletter={-,>>,+,&,|,_},comment=[l][\footnotesize\sffamily\textbf]{\!}
}

\lstdefinelanguage{p4}
{keywords=[2]{control, action, else, if, table, blackbox, register, field_list, field_List, field_list_calculation},
 otherkeywords={register_read, modify_field, apply, bit_or, register_write, input, algorithm, output\_width, width, default\_action, size, instance_count, exact,lpm,reads,actions  ,bit_and, clone_ingress_pkt_to_egress, add_header}, 
sensitive=true, alsoletter={-,>>,+,&,|,_},
basicstyle=\color{red}\ttfamily,
keywordstyle=\color{darkgreen}\ttfamily,
keywordstyle=[2]\ttfamily\bfseries\color{steel_blue},
commentstyle=\ttfamily, 
stringstyle=\ttfamily, 
identifierstyle=\ttfamily,
alsoletter={0,1,2,3,4,5,6,7,8,9}
comment=[l][\footnotesize\ttfamily\color{red}]{\!}
}

\lstdefinelanguage{query2}
{keywords={MapBolt, ReduceBolt, >>, >+, +, &, push, if_, elif_, else, pop, mapInit, 
match, fwd, modify, set, mod, sample, sampleD, sampleS, $\triangleleft$, announce, 
withdraw, mapD, mapS, map, reduce, filter, filterD, filterS, runningReduceD, runningReduceS, 
distinct, toList, mapValues, countS, countByWindow, countByValueAndWindow, reduceByKey, window,
transform, map-init, update-metadata, update-headers, emit, join}, 
keywordstyle=\bfseries\ttfamily,
keywordstyle=[2]\ttfamily\bfseries,
commentstyle=\ttfamily, 
stringstyle=\ttfamily, 
identifierstyle=\ttfamily,
emph={trafficAnomalyIPs}, 
emphstyle=\ttfamily\bfseries\color{red},
sensitive=true, alsoletter={0,1,2,3,4,5,6,7,8,9,-,>>,+,&,|,_},comment=[l][\footnotesize\sffamily\textbf]{\!}
}

\lstdefinelanguage{query3}
{keywords={MapBolt, ReduceBolt, >>, >+, +, &, push, if_, elif_, else, pop, mapInit, 
match, fwd, modify, set, mod, sample, sampleD, sampleS, $\triangleleft$, announce, 
withdraw, mapD, mapS, map, reduce, filter, filterD, filterS, runningReduceD, runningReduceS, 
distinct, toList, mapValues, countS, countByWindow, countByValueAndWindow, reduceByKey, window,
transform, map-init, update-metadata, update-headers, emit, join}, 
keywordstyle=\bfseries\ttfamily,
keywordstyle=[2]\ttfamily\bfseries,
commentstyle=\ttfamily, 
stringstyle=\ttfamily, 
identifierstyle=\ttfamily,
emph={map}, 
emphstyle=\ttfamily\bfseries\color{red},
sensitive=true, alsoletter={0,1,2,3,4,5,6,7,8,9,-,>>,+,&,|,_},comment=[l][\footnotesize\sffamily\textbf]{\!}
}

\lstdefinelanguage{Scala}%
  {morekeywords={abstract,case,catch,class,def,%
    do,else,extends,false,final,finally,%
    for,if,implicit,import,lazy,match,mixin,%
    new,null,object,override,package,%
    private,protected,requires,return,sealed,%
    super,this,trait,true,try,%
    type,val,var,while,with,yield},
otherkeywords={=,=>,<-,<\%,<:,>:,\#,@},%
   sensitive,%
   morecomment=[l]//,%
   morecomment=[n]{/*}{*/},%
   morestring=[b]",%
   morestring=[b]',%
   morestring=[b]""",%
  }[keywords,comments,strings]%

\usepackage{adjustbox}
\usepackage{lipsum}
\newcommand{\smartparagraph}[1]{\noindent{\bf #1}\ }
\usepackage{amsthm}

\newtheorem{theorem}{Theorem}

\usepackage{collcell}
\usepackage{hhline}
\usepackage{pgf}
\usepackage{multirow}

\newcolumntype{R}[2]{%
    >{\adjustbox{angle=#1,lap=\width-(#2)}\bgroup}%
    l%
    <{\egroup}%
}

\def\colorModel{RGB} 

\newcommand\ColCell[1]{
  \pgfmathparse{#1}
  \pgfmathtruncatemacro\resultr{23+#1}
  \pgfmathtruncatemacro\resultg{55+#1}
  \pgfmathtruncatemacro\resultb{94+#1}

   \ifnum\pgfmathresult=0
        \relax\color{white}
        \pgfmathsetmacro\compA{255}      
        \pgfmathsetmacro\compB{255} 
        \pgfmathsetmacro\compC{255} 
    \else
        \definecolor{foo}{RGB}{\resultr,\resultg,\resultb}  
        \relax\color{foo}
        \pgfmathsetmacro\compA{\resultr}      
        \pgfmathsetmacro\compB{\resultg} 
        \pgfmathsetmacro\compC{\resultb} 
    \fi
  \edef\x{\noexpand\centering\noexpand\cellcolor[\colorModel]{\compA,\compB,\compC}}\x #1
  } 
\newcolumntype{E}{>{\collectcell\ColCell}m{0.4cm}<{\endcollectcell}}  

\newcommand{\supsym}[1]{\raisebox{4pt}{{\footnotesize #1}}}

\newcommand{\ucsb}{\supsym{$\star$}}
\newcommand{\wpa}{\supsym{$\diamond$}}
\newcommand{\nik}{\supsym{$\ddag$}}

\begin{document}
\date{}
\title{\system: Planning for Dynamics in Network Streaming Analytics Systems}
\author{
{Rohan Bhatia\ucsb, Arpit Gupta\ucsb, Rob Harrison\wpa, Daniel Lokshtanov\ucsb, Walter Willinger\nik}\\
\ucsb\normalsize{UC Santa Barbara}~~~\wpa\normalsize{United States Military Academy}~~~\nik\normalsize{NIKSUN Inc.}\\
}

\input{abstract}

\maketitle
\pagestyle{plain}



\begin{sloppypar}
\input{intro-rev1}
\input{background-vers2}
\input{boot}
\input{runtime}

\input{predict}

\input{implementation}

\input{evaluation}

\input{related-new}
\input{conclusion}

\begin{acks}
The views expressed in this paper are those of the authors and do not reflect the official policy or position of the U.S. Military Academy, the Department of the Army, the Department of Defense, or the U.S. Government.
\end{acks}
\label{lastpage}

\clearpage
\pagebreak
\small
\balance\bibliographystyle{acm}
\bibliography{paper,vighata}
\clearpage
\normalsize
\input{appendix}

\end{sloppypar}


\end{document}

%% file: abstract.tex
\begin{abstract}   
The emergence of programmable data-plane targets has motivated a new hybrid design for network streaming analytics systems that combine these targets' fast packet processing speeds with the rich compute resources available at modern stream processors. However, these systems require careful {\em query planning}; that is, specifying the minute details of executing a given set of queries in a way that makes the best use of the limited resources and programmability offered by data-plane targets. 
Most hybrid systems do not support query planning. The ones that do, such as Sonata, employ static query planning. Using real-world packet traces we show that static query plans cannot handle even small changes in the input workload, wasting data-plane resources to the point where query execution is confined mainly to userspace. 

This paper presents the design and implementation of \system, a new network streaming analytics system that employs dynamic query planning to deal with the dynamics of real-world input workloads. Specifically, we develop a suite of practical algorithms for (i) computing effective initial query plans (to start query execution) and (ii) enabling efficient updating of portions of such an initial query plan at runtime (to adapt to changes in the input workload). Using real-world packet traces as input workload, we show that compared to existing systems, such as Sonata, \system reduces the stream processor's workload by more than two orders of magnitude.


\end{abstract}

%% file: intro-rev1.tex
\section{Introduction}
\label{sec:intro}

Networks supporting Internet-connected applications and devices have witnessed exponential growth in both scale and complexity. To effectively manage these networks, operators must keep track of both coarse and fine-grained aspects of network traffic from a variety of perspectives, \eg device-specific, service-specific, and aggregate. They use this information to detect and, ultimately, react to an increasingly diverse set of performance and security-related network events or requirements in near-real-time. Unfortunately, the number of required tasks and the speed at which they need to be executed exceed the capabilities of human operators and have motivated many network service providers to consider taking humans ``out-of-the-loop" altogether through automation. To achieve this objective, today's networks rely increasingly on modern network streaming analytics systems capable of running a large number and wide range of different queries at line rate for up to billions of packets per second. 

A common feature of many of these systems leverages programmable data plane targets (\eg Intel Tofino~\cite{tofino}), which are capable of executing some or all parts of different queries at line rate.  These devices can also send intermediate results to state-of-the-art stream processors (\eg Apache Flink~\cite{flink}) for further back-end processing. Most of these hybrid systems (\eg Marple~\cite{marple}, Newton~\cite{newton}, Sonata~\cite{sonata}, UnivMon~\cite{univmon}, etc.) specify how to break down the input queries into individual dataflow operators, how to compile these operators in the data plane, and how to combine the output of respective operators executed in the data plane with modern stream processors to answer various queries accurately. However, since modern programmable data planes are inelastic by design and highly constrained in terms of available hardware resources (\eg memory, ALUs), these systems also need to decide how to partition the dataflow operators for each query between the data-plane target and the stream processor and how to allocate resources for the different dataflow operators in the data plane. This decision-making is often called {\em query optimization} or {\em query planning}, and its efficacy dictates a systems' ability to scale to higher data rates and a larger number of input queries. 

Unfortunately, most of the existing hybrid systems only focus on the mechanisms of using data-plane targets for query execution, and do not consider query planning. For example, in the case of Marple~\cite{marple}, it is unclear how to configure the data-plane targets to support multiple queries for arbitrary traffic workloads. In contrast, Sonata~\cite{sonata} supports query planning by deciding how to iteratively refine the input queries such that data-plane targets selectively spend limited resources on the traffic of interest (iterative refinement); how to map the queries' dataflow operators to stateful/stateless data structures in the data plane (operator mappings); and how to allocate resources to these data structures (resource allocation). However, Sonata's query planning is ``static" or ``early-binding" in nature; query planning decisions are made at compile time and not at runtime.

Though these ``early bindings" minimize runtime complexity, they prevent the system from reacting to changes in traffic patterns over time and, therefore, cause the system to miss out on possible opportunities to reap scalability gains. This observation begs the question: can designing hybrid network streaming analytics systems based on more dynamic query planning improve these systems' scalability, and if so, how much?

\begin{figure}[t] 
\begin{center}
\includegraphics[scale=0.6]{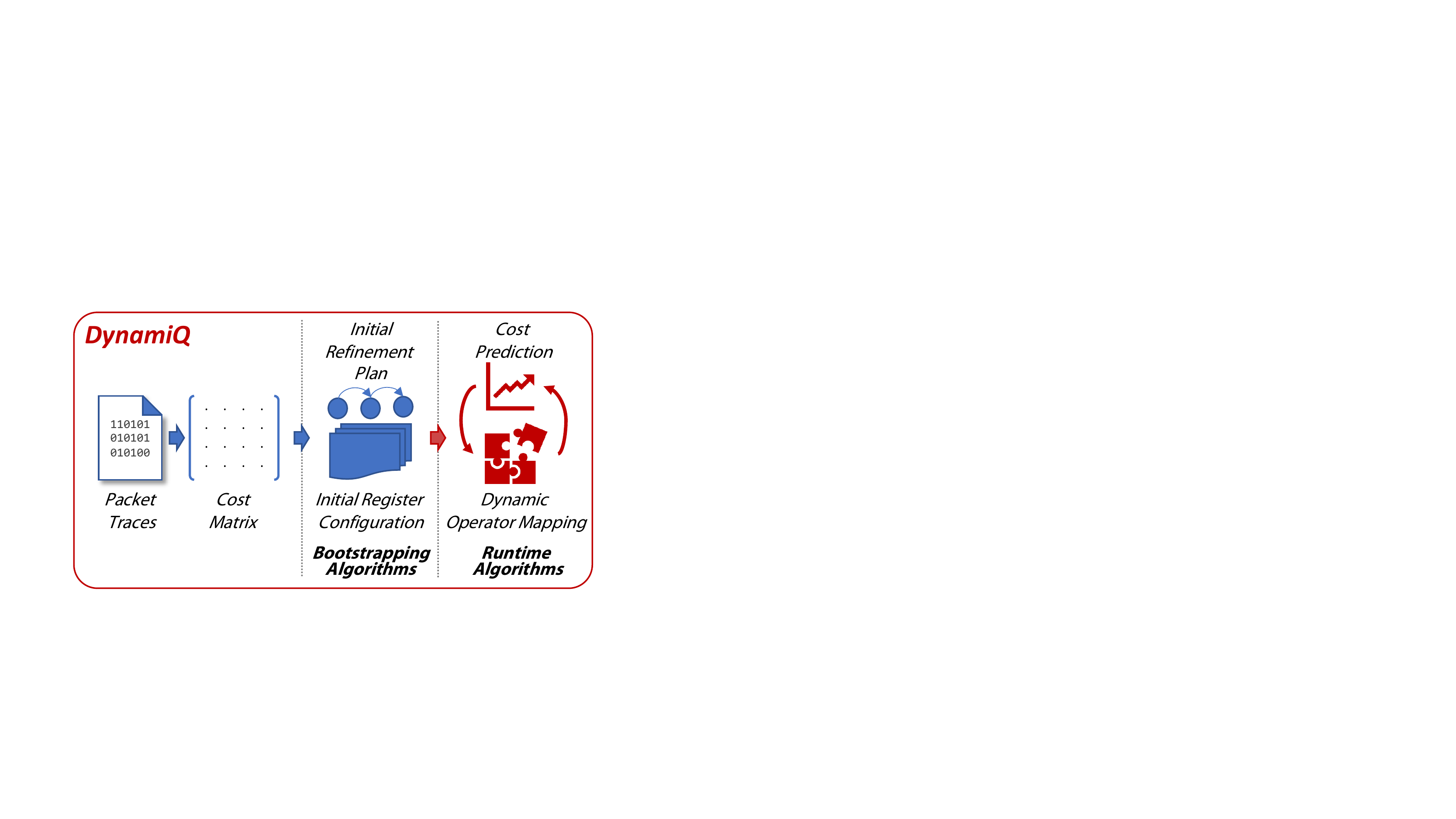}
\caption{\system's architecture for dynamic query planning and execution.}
\label{fig:sonata-dynamiq}
\vspace{-.25in}
\end{center}
\end{figure}

In this paper, we answer that question in the affirmative by developing new query planning algorithms that combine query-specific iterative refinement with dynamic query planning and exploit the opportunities afforded by this unique combination. We can divide these algorithms into two categories: (1)~\textit{bootstrapping algorithms} for computing an initial query plan that is amenable to dynamic updates at compile-time, and (2)~\textit{runtime algorithms} for dynamically updating this initial query plan at runtime. To demonstrate their feasibility in practice, we integrate these algorithms into a Tofino-based prototype implementation {\bf \system}. Figure~\ref{fig:sonata-dynamiq} gives an overview of \system's architecture. This work makes the following concrete contributions:

\smartparagraph{Traffic Dynamics Analysis (Section~\ref{sec:background}).} 
We leverage packet-level traffic traces from real-world production networks to better understand static query-planning's limitations. In particular, we consider Sonata~\cite{sonata} and use a set of telemetry queries (Table~\ref{tab:queries}) to show that relying on static query plans can increase the amount of traffic sent to the stream processor by 3-4 orders of magnitude---the same amount of work produced by offloading query execution to the data plane in the first place. 
We also provide evidence that dynamically redistributing the memory allocated to different operators at runtime suffices to adapt to input workload dynamic---motivating \system's query-planning algorithms.


\smartparagraph{Bootstrapping algorithms (Section~\ref{sec:bootstrap}).} 
Recognizing that it is not possible to change query refinement decisions and some data-plane resource allocations at runtime (\eg size of registers for stateful operations), computing effective refinement plans and data-plane resource allocations at compile-time is critical.  Additionally, we must compute the same to be amenable to some runtime adjustment.  To this end, we present the design and implementation of two algorithms that compute initial query refinement plans and data-plane resource allocations, respectively. We demonstrate that these algorithms have desirable properties, including the ability to find effective query plans for different input workloads and data-plane targets.

\smartparagraph{Runtime algorithms (Section~\ref{sec:runtime}).} 
We present the design and implementation of two runtime algorithms for adapting to traffic dynamics. The first algorithm considers an operator's past memory requirements and estimates that operator's required memory for the future window. The second algorithm uses these predicted values to compute new operator mappings in the data plane dynamically.  We demonstrate that this operator-mapping problem is NP-Hard and that, when implemented, these algorithms introduce modest overhead.

In Section~\ref{sec:implement}, we describe how \system enables dynamic operator mapping using Intel's Tofino-based switch~\cite{tofino}. In Section~\ref{sec:eval}, we use a set of realistic monitoring queries on publicly-available, real-world packet traces~\cite{caida15} to show that \system's query planning achieves near-optimal performance, reducing the stream processor's workload by more than \textbf{two orders of magnitude}, compared to the static query planning used by Sonata.
We also demonstrate that the results are robust to input traffic and query workload, switch constraints, etc. We will make our prototype, code base and datasets publicly available to support additional research (Section~\ref{sec:future}).

This work does \textbf{not} raise any \textbf{ethical issues}.

%% file: background-vers2.tex
\section{Background and Motivation}
\label{sec:background}
In this section, we first describe how hybrid telemetry systems use query planning to compute static query plans. We then concretely explore the limitations of static query planning and argue for the need to develop dynamic query plans. We make use of the publicly-available code base of Sonata~\cite{sonata-github} to make these points.

\subsection{Network Streaming Analytics Systems}
The emergence of high-speed, programmable data-plane switches has motivated new, hybrid designs for network monitoring systems that combine these devices' fast packet processing speeds with the flexible and horizontally scalable data processing capabilities of modern stream processors (\eg Apache Spark~\cite{spark}). Such designs are especially desirable given the current trend towards automating network management and increasingly more complex management tasks requiring more fine-grained levels of monitoring and control. 



At the same time, these systems must cope with the limitations inherent in coupling programmable data plane targets with general-purpose CPUs. On the one hand, the former have much faster packet processing speed (in the Tbps range) than the latter, but their programmability is limited compared to general-purpose CPUs. On the other hand, data plane targets have limited compute resources (\eg memory, processing stages, storage) that host-based CPU systems do not face. Further advances in data plane technology notwithstanding, the compute resources available to programmable data plane targets are considered precious and must be used judiciously for the foreseeable future. 


Given the scarce resources available in the data plane, hybrid, network streaming analytics systems must carefully perform {\em query planning}; that is, given a set of network monitoring queries, the system must determine which portions of each query to execute where without compromising the queries' accuracy. Unlike the database research literature (\eg see~\cite{hellerstein2017} and references therein), query planning has received little attention to date in the networking research literature and is largely absent from most existing network streaming analytics system (\eg Marple~\cite{marple}, UnivMon~\cite{univmon} and others). However, we expect effective query planning to play a decisive role in determining the practicality of deploying these hardware-based network streaming analytics systems in production networks.

\begin{figure}[t!]
\begin{lstlisting}[language=query2,basicstyle=\footnotesize, 
basicstyle=\footnotesize, numbers=left,xleftmargin=2em,frame=single,
framexleftmargin=2.0em, captionpos=b, label=superspreader, caption=Detect superspreader hosts. Executing this query at at a coarser level (e.g. 8) entails adding a {\tt map} operator that sets the last 24 bits of {\tt sIP} to zero before line~2.]
packetStream(W)
.map(p => (p.sIP, p.dIP))
.distinct()
.map((sIP, dIP) => (sIP, 1))
.reduce(keys=(sIP,), f=sum)
.filter((sIP, count) => count > Th)
\end{lstlisting}
\ifx \compress \undefined
\else
\vspace*{-1\baselineskip}
\fi
\end{figure}

\subsection{Case Study: Static Query Planning}
\smartparagraph{An illustrative query as running example.} Query~\ref{superspreader} shows how network operators can use Sonata~\cite{sonata} to express the superspreader detection task as a dataflow query over a packet stream. Under the hood, Sonata's query planner divides a representative input workload in the form of packet traces into consecutive windows (\eg $3$ seconds).  The planner first seeks to perform query refinement; this is a process of iteratively ``zooming-in'' over successive windows on portions of the traffic that satisfy the input query at a coarser level of granularity. Planning for query refinement boils down to exploring all refinement levels that the network operator allows and choosing the set of refinement levels that best reduce the workload for the stream processor without compromising accuracy. Refinement levels are based on the hierarchical structure inherent to specific keys in a query.  For example in Query~\ref{superspreader}, source and destination IP addresses ({\tt sIP} and {\tt dIP}, respectively) inherently possess hierarchy.  The finest refinement level would be an IP address with a {\tt/32}-bit mask applied and all those $(n<32)$-bit masks would be coarser levels.  Of many, one candidate refinement plan could be: {\tt *} $\rightarrow$ {\tt 8} $\rightarrow$ {\tt 32}.  In the first window, the {\tt /8}-bit masked version of Query~\ref{superspreader} executes by modifying the original query and inserting {\tt map(sIP => sIP/8)} before line 2.  In the immediately following window, the original input query, \ie the {\tt /32}-bit masked version of Query~\ref{superspreader}, is then executed for those source IP addresses that satisfied the query at the {\tt /8} refinement level. Note that query refinement reduces the amount of data plane resources required per-window at the cost of returning an answer to the original query with a slight delay; this delay is a function of the selected monitoring window and the length of the refinement plan. After selecting a refinement plan and augmenting Query~\ref{superspreader} for the selected refinement levels, Sonata's query planner next focuses on query partitioning; it decides for each of the refined queries how to execute which subset of dataflow operators (\eg {\tt distinct}, {\tt reduce}, etc.) in the data plane. To make this decision, Sonata's query planner must map specific dataflow operators to data plane resources and how much of the finite switch memory to allocate for each operator. 






\smartparagraph{From packet traces to cost matrices.}
To compute this optimal query plan, Sonata's query planner 
generates for each window a cost matrix by synthesizing and executing versions of each refined query. Specifically, the cost matrix provides estimates of the amount of memory required 
and the number of packet tuples processed 
by each of the stateful operators used in each combination of all considered refinement and partitioning plans. If $Q$, $L$, and $O_q$ denote the sets of queries, refinement levels, and stateful operators for query $q \in Q$, respectively, then the total number of queries Sonata has to run to generate the cost matrix for one window is on the order of 
$\sum_{q \in Q}{|L| \choose 2}*|O_q|$. 
In the case of Query~\ref{superspreader}, $|O|$=2 ({\tt distinct} and {\tt reduce}) and we will assume a network operator-specified maximum refinement levels of $|L|$=8. With these parameters, $56$ different cost matrix values must be computed and evaluated just for Query~\ref{superspreader}. Clearly, computing cost matrices is expensive and, even with state-of-the-art data analytics systems (\eg BigQuery~\cite{bigquery}), computing them for one hour of packet traces takes around four hours.



\smartparagraph{From cost matrices to query plans.}
Sonata's query planner then provides the cost matrices as input to an integer linear program (ILP). This ILP also takes various data-plane and ordering constraints into consideration when computing the query plan that minimizes the number of packets that are sent to the stream processor. More precisely, Sonata considers a Protocol Independent Switch Architecture(PISA)-based data-plane target\cite{pisa}, whose packet-processing pipeline is divided into multiple physical stages with finite resources. For stateful operations, this target uses stateful ALUs and SRAM memory, called registers, local to a single stage to update and maintain state, respectively. Packets use metadata fields to carry additional information, such as, intermediate query results, across different stages. Sonata's ILP formulation incorporates the following data plane constraints: amount of metadata that can be stored ($M$), the number of stateful ALUs per stage ($A$), the amount of stateful memory in bits per stage ($B$), and the number of stages ($S$). Solving this ILP yields a static query plan consisting of: (1)~{\em refinement plan}, (2)~{\em operator mapping}---between dataflow operators and data-plane components, and (3)~{\em register allocation}---deciding how much of the scarce SRAM memory to allocate to each register for stateful operations in the data plane.

\begin{figure}[t] 
\begin{minipage}{1\linewidth}
\begin{subfigure}[b]{.49\linewidth}
\includegraphics[width=\linewidth]{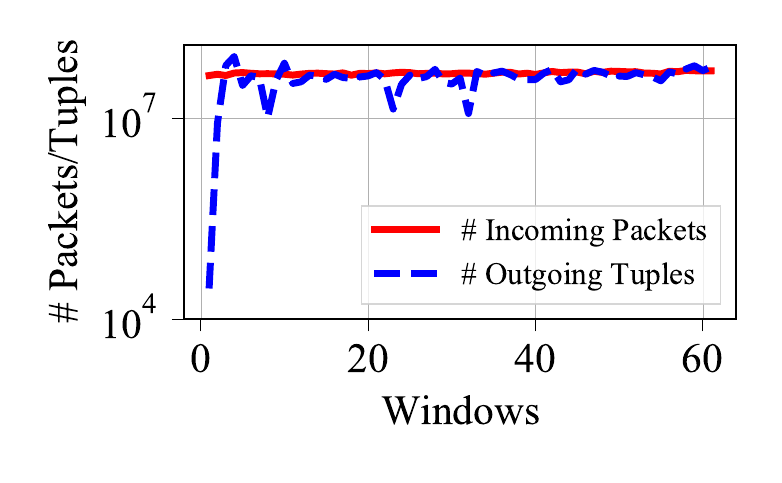}
\caption{Input vs. SP workload}
\label{fig:sp-workload}
\end{subfigure}
\begin{subfigure}[b]{.49\linewidth}
\includegraphics[width=\linewidth]{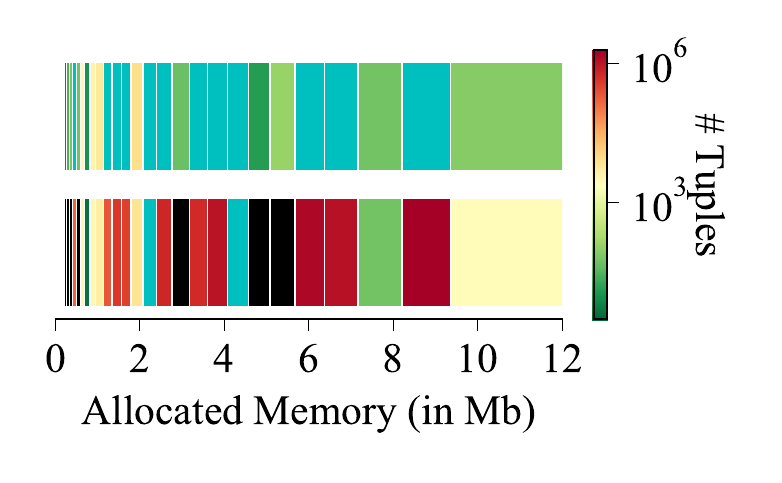}
\caption{Memory alloc. and load}
\label{fig:qp_dynamics}
\end{subfigure}
\end{minipage}
\caption{Cost of not adapting. SP is short for stream processor in (a). 
}
\label{fig:cost_of_not_adapting}
\end{figure}

\subsection{Limitations of Static Query Plans}
\label{ssec:static-harmful}

For a given set of queries, Sonata assumes that neither the input workload (\ie total number of packets arriving at the switch) nor the resulting workload for the computed query plan (\ie cost matrices) change much or at all over time. As a result, Sonata's static query plan that was ``trained'' using historical packet traces from the network is assumed to also perform well in the future.

\smartparagraph{Wasteful use of hardware resources.}
However, an empirical evaluation of Sonata's performance shows that this assumption does not hold in practice. In our evaluation, we use the set of input queries described in~\cite{sonata} and CAIDA’s anonymized and unsampled packet traces~\cite{caida15} as input workload (see Section~\ref{sec:eval} for details). Figure~\ref{fig:sp-workload} depicts the number of packets observed at the switch and the resulting tuples received at the stream processor over 60 successive three-second windows. For this experiment, Sonata's query planner was trained on the input workload data from the first window only to highlight the perils of using a stale query plan.  Though the query plan effectively reduced the number of tuples processed at the stream processor (blue line) from the number of packets processed at the switch by several orders of magnitude, using this stale query plan in successive windows resulted in an almost instantaneous workload increase at the stream processor. After the first window, the number of tuples processed at the stream processor is nearly in parity with the total input workload arriving at the switch (red line). Using a stale query plan trained on different windows as input shows similar behavior.  Relying on Sonata'a static query planning can eliminate the benefit of partially executing query operators in the data plane and effectively relegates query execution to user space -- similar to not having partitioned query execution to a programmable data plane in the first place.



\smartparagraph{Why are static query plans ineffective?} 
To explain the root cause for why static query plans fail to reduce the number of tuples sent to the stream processor, we first examine how data plane resources are consumed with a stale query plan.  If a static query plan works well for a given window ($w$), there are three possible outcomes for each of the stateful operators in the subsequent window ($w'$): the amount of memory required to execute the stateful operator in the data plane in window $w'$ is (1) less than, (2) more than, or (3) exactly the same as that required in window $w$. Only in this last case does static query planning have a clear advantage.  The first case results in under-utilizing limited switch memory and introduces an opportunity cost---the wasted memory could have been allocated for other operators that need more memory in $w'$. In the second case, all operators that Sonata cannot fit into the switch's finite memory cannot be fully executed in the data plane. This case is even more consequential for queries with more than one stateful operator; should there not be enough memory for the first operator to execute in the data plane, then all the packets from the first operator onward will be sent to the stream processor (see Appendix~\ref{ssec:query_acc}). Nevertheless, case (2) also affords an opportunity cost---any stateful memory that had been allocated to any of the subsequent operators will go unused and could have been allocated to other operators that would benefit from more memory in $w'$. 

Figure~\ref{fig:qp_dynamics} visualizes how Sonata's static query plan fares over two consecutive windows $w$ (top) and $w'$ (bottom). Each window is shown as a horizontally-laid-out, stacked bar plot, 
where each block in a stacked bar plot represents a stateful operator; the block sizes are proportional to the amount of allocated memory as specified by the static query plan. The color signifies that operator's contribution to the number of tuples sent to the stream processor. Here, a cyan-colored block represents the first stateful operator whose output is processed in the data plane itself resulting in no contribution to the workload at the stream processor. On the other hand, a black-colored block represents the second (or some later) stateful operator that cannot be executed in the data plane because the first (or some previous) operator could not be allocated sufficient memory. 
We observe that $46\%$ of all operators in this static query plan are over-provisioned (case~1) and $44\%$ are under-provisioned (case~2). This analysis highlights static query plans' inability to adapt to variation in operator memory requirements from one window to the next. This inability leads to wasted data plane resources and causes a much higher workload at the stream processor for future windows compared to the window for which the static query plan was trained.

\smartparagraph{Alternative static query planning techniques.}
We also considered alternative static query plans that enhance Sonata's query planner to more effectively handle workload dynamics. A detailed discussion of our explorations of these alternative static query planning techniques is presented in Appendix~\ref{sec:alt_static}. The main takeaway from these investigations is that while some of these techniques perform better than Sonata's solution, the gains are only marginal and highly sensitive to the input workload.

\begin{figure}[t] 
\begin{minipage}{1\linewidth}

\begin{subfigure}[b]{.49\linewidth}
\includegraphics[width=\linewidth]{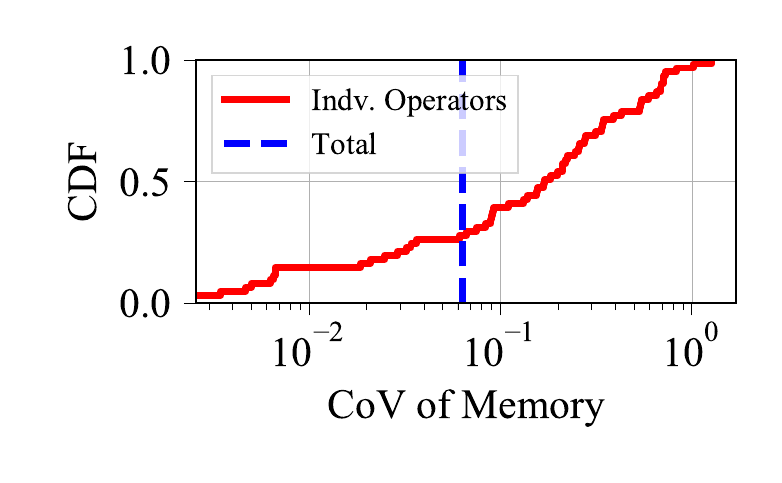}
\caption{CoV of required memory}
\label{fig:agg_vs_individul_mem}
\end{subfigure}
\begin{subfigure}[b]{.49\linewidth}
\includegraphics[width=\linewidth]{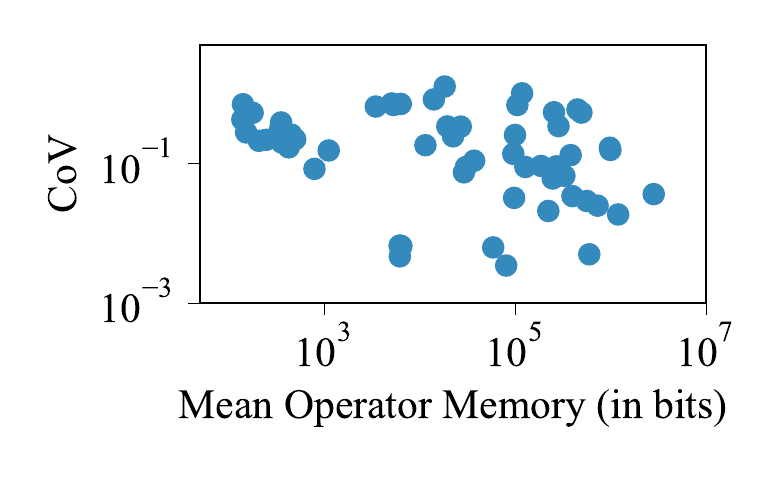}
\caption{CoV vs. mean memory}
\label{fig:cov_distr_mem}
\end{subfigure}
\end{minipage}
\caption{
Opportunities for dynamic query plans. 
}
\label{fig:observations-reg-size}
\end{figure}

\subsection{Dynamic Query Planning to the Rescue}
Static query plans fail to adapt to the variability in the operators' memory requirements over time and are therefore a major roadblock towards effectively leveraging programmable data plane targets in support of network streaming analytics. To overcome the weaknesses of static query planning, we argue here for considering dynamic query planning and provide a short preview of how we incorporate dynamic query planning into the design of our new hybrid network streaming analytics system, \system.

\smartparagraph{Redistributing memory at runtime.}
Figure~\ref{fig:qp_dynamics} shows two successive windows to visualize the prevalence of under- and over-provisioning of operator memory that results from static query planning; Figure~\ref{fig:agg_vs_individul_mem} further illustrates this point. For each of the stateful operators in Sonata's static query plan, we quantify the variability of required memory across sixty consecutive windows using the coefficient of variance (CoV) and show the distribution of obtained CoV-values (red line). We observe that the CoV-values for the different operators exhibit significant variability, especially when compared with the low variability of the total memory required (blue line). Moreover, Figure~\ref{fig:cov_distr_mem} shows that high variability is not restricted to the smaller operators but that  operators with a high memory footprint can also exhibit high variability. The low variability of the total memory combined with the high variability of the memory required for individual operators suggests that, for any window, the memory requirements for some operators increase while they decrease for the others. This observation also suggests that dynamically redistributing operators at runtime should more efficiently use available switch memory such that the operators that require more/less memory are assigned more/less resources.


\smartparagraph{Optimal, yet impractical, dynamic query planning.}
To dynamically redistribute operator memory at runtime, one approach is to use Sonata's query planner and compute a new query plan for each window separately. Although such an approach will output optimal query plans that minimize the load at the stream processor, the obtained plans are  impractical. First, implementing these plans requires resizing register memory in the data plane. However, for most existing data-plane targets, resizing data-plane registers implies recompilation, meaning taking the switch offline for a few seconds\cite{newton}. Second, realizing these plans requires changing the queries' refinement plans at runtime which in turn adversely affects the queries' accuracy (see Appendix~\ref{ssec:query_acc} for details). Lastly, computing these plans requires knowing each window's input data and resulting cost matrix \emph{prior} to the beginning of a window---an impossible task.

\smartparagraph{Our approach in a nutshell: \system.}
Given the opportunities that careful query planning affords and the practical constraints that modern data plane targets impose, we consider in this paper an approach to dynamic query planning where instead of changing the entire query plan, we only update operator mappings at runtime. In particular, our approach consists of fixing the refinement plan and register size configuration at compile time prior to running the queries and then dynamically changing the mapping between operators and data-plane registers at runtime. We show that such an approach efficiently uses the scarce data plane resources and can be implemented at runtime.  First, we must compute both an initial refinement plan and an initial allocation of register sizes in a way that facilitates dynamically mapping operators to registers at runtime.  We use the term {\em bootstrapping} to describe the process of determining these initial refinement plans and register size allocations.  We then show that bootstrapping can be done in ways that are less sensitive to input workload dynamics and more dependent on the structure of the query workload, which we assume in this paper to be invariant over time (Section~\ref{sec:bootstrap}). To update the operator mappings at runtime, we need a learning model to predict the cost matrix for future windows. We then use these predicted values to compute new operator mappings in ways that incur minimal overheads (Section~\ref{sec:runtime}). 

%% file: boot.tex
\section{Bootstrapping Query Plans}
\label{sec:bootstrap}
Given that runtime changes of register sizes in the data plane and refinement plans are not feasible, we consider a purposefully-designed compile time ``bootstrapping" effort. This effort consists of computing an initial refinement plan (Section~\ref{ssec:ref_plan}) and initial register size allocation (Section~\ref{ssec:reg_sizes}) at compile time in ways that ultimately facilitate efficient operator mappings (\ie memory redistribution) at runtime (Section~\ref{sec:implement}).

\subsection{Computing Refinement Plan}
\label{ssec:ref_plan}


\smartparagraph{Total operator memory.}
We quantify a refinement plan's memory footprint using the metric of {\em total operator memory (TOM)}, which is defined as the sum of memory required for all the stateful operators of a given refinement plan. Given the limits on the total available memory and the overall number of ALUs in the data plane, when selecting an effective initial refinement plan from among all considered plans, our goal should be to select one that has the smallest TOM-value and the fewest total number of stateful operators. To illustrate, in the case of our example Query~\ref{superspreader}, for the earlier considered refinement plan {\tt *} $\rightarrow$ {\tt 8} $\rightarrow$ {\tt 32}, the total number of stateful operators is four. Table~\ref{tab:cost-matrix-q1} shows that the memory requirements for the two stateful operators at the {\tt *} $\rightarrow$ {\tt 8} and {\tt 8} $\rightarrow$ {\tt 32} levels are ($523.7$~K, $948.8$~K) and ($6.5$~K, $14$~M), respectively. Therefore, this plan's TOM value is $0.52 + 0.95 + 0.0065 + 14 \approx 15.5$ M. Also note that while the alternative refinement plan {\tt *} $\rightarrow$ {\tt 16} $\rightarrow$ {\tt 32} has the same number of stateful operators, its TOM-value is about 50\% lower (\ie $7.2$ M). Therefore, we would prefer to use the refinement plan {\tt *} $\rightarrow$ {\tt 16} $\rightarrow$ {\tt 32} as it has a lower memory overhead. 


\begin{figure}[t] 
\begin{minipage}{1\linewidth}
\begin{subfigure}[b]{.49\linewidth}
\includegraphics[width=\linewidth]{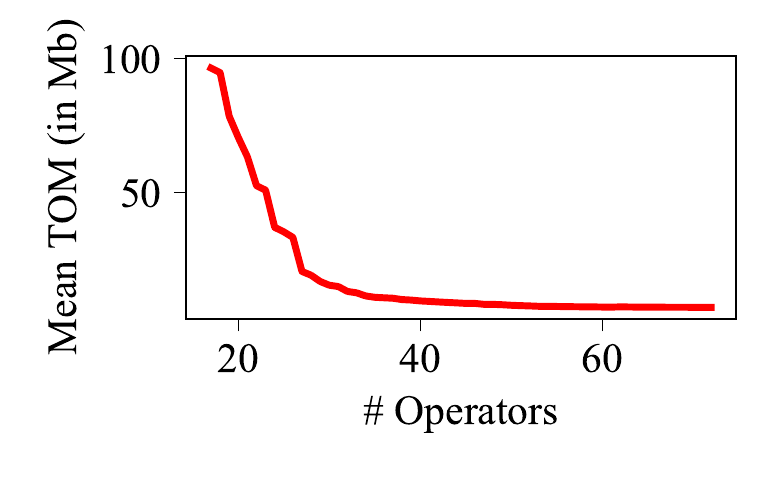}
\caption{Min. Mean TOM}
\label{fig:max_tom_vs_max_total_ops}
\end{subfigure}
\begin{subfigure}[b]{.49\linewidth}
\includegraphics[width=\linewidth]{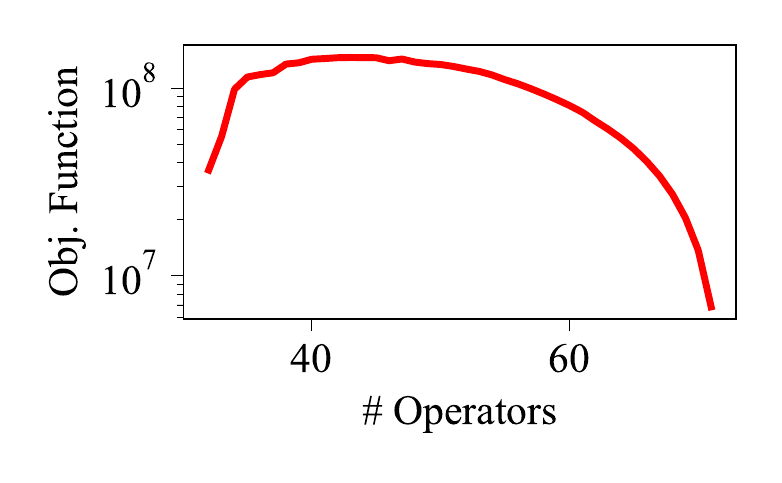}
\caption{Objective function}
\label{fig:heuristic_function}
\end{subfigure}
\end{minipage}
\caption{
Relationship between TOM and total operators
}
\label{fig:tomtomtom}
\end{figure}

\smartparagraph{Relationship between TOM and total number of operators.}
To characterize the relationship between TOM and the number of operators, we use the cost matrices for the queries considered in Section~\ref{sec:background} and performed the following experiment. For a given number of operators ($k$), we consider for each window the refinement plans that are pertinent to these $k$ operators and select the refinement plan for which the mean of all TOM-values across all windows is minimum (min-mean-TOM). We repeat this experiment for each $k$, starting with the minimum value that is determined by the number of operators required to execute all queries at the finest refinement level (i.e., $\slash$32) and ending with the maximum value that is configuration-specific and is capped by the number of available ALUs (registers) in the data plane target used (\ie 72).

Figure~\ref{fig:max_tom_vs_max_total_ops} shows that the min-mean-TOM values decrease as the number of operators ($k$) increases. It shows that refinement plans with smaller TOM require more operators (and vice versa) and argues for selecting a refinement plan that strikes the right balance between the total operator memory and the number of operators. 

\begin{table}[t]
\begin{footnotesize}
\begin{center}
\begin{tabular} {|c|c|c|c|c|}
\hline
\multicolumn{5}{|l|}{\textbf{Distinct operator}} \\
\hline
\textit{refs} & 8 & 16 & 24 & 32\\
\hline
{\tt *} & 523.7 K & 727.3 K & 954.9 K & 1.2 M \\
8 & - & 572.5 K & 766.1 K & 948.8 K \\
16 & - & - & 384.7 K & 521.9 K \\
24 & - & - & - &  244.3 K \\
\hline
\hline
\multicolumn{5}{|l|}{\textbf{Reduce operator}} \\
\hline
\textit{refs} & 8 & 16 & 24 & 32\\
\hline
{\tt *} & 6.5 K & 596.6 K & 6.7 M & 18.6 M \\
8 & - & 232.3 K & 4.1 M & 14 M \\
16 & - & - & 267.4 K & 5.4 M \\
24 & - & - & - &  93.3 K \\
\hline
\end{tabular}%
\end{center}
\end{footnotesize}
\caption{Cost matrix for Query 1 showing $B$ values}
\label{tab:cost-matrix-q1} 
\ifx \compress \undefined
\else
\vspace*{-1\baselineskip}
\fi
\end{table}

\smartparagraph{TOM-based heuristic.}
We next leverage our empirically-confirmed intuitive argument to develop a simple heuristic algorithm for selecting an initial refinement plan. To this end, for each $k$ (\ie number of operators), our algorithm first computes the refinement plans for which the mean TOM-values across all input windows is minimum. It then searches for the $k$-value that maximizes the objective function $U(m,o)=mo$ where $m$ denotes the difference between the total available memory in the data plane and the computed value of min-mean-TOM and $o$ is the difference between the available number of registers and number of operators.
In effect, this algorithm is determining a refinement plan for which the difference between the available and required memory and between the available and required registers is as high as possible. Figure~\ref{fig:heuristic_function} shows the convex shape of this objective function as a the number of operators varies across the range of possible values.



\subsection{Computing Register Sizes}
\label{ssec:reg_sizes}


Our goal here is to compute (optimal) register sizes that given the refinement plan enable efficient operator mappings at runtime. This problem can be formulated as an ILP, however the performance of the solvers on the generated ILPs is unsatisfactory, possibly due to the NP-hardness of solving general ILPs. Additionally, there is the risk that ILP solutions would be over-fitted on the input data. As a result, we look for alternative solutions in the form of simple heuristics that are ``good enough". However, developing heuristic solutions for this problem is further complicated by hardware-specific constraints on available data-plane resources. Specifically, there is an upper bound on the memory allocated to an individual register. There is also a cap on the total memory allocated to all the registers in each physical stage. 

\smartparagraph{Right-sizing a single stage.} Since finding a memory allocation that works for all physical stages simultaneously, satisfies the various constraints, and helps us achieve the stated objective for our bootstrapping effort is challenging, we consider the following simplified version of the problem. The simplification consists of decomposing the problem by first finding a register-size configuration that works well for a single stage and then replicating that same configuration across all stages. Here, our intuition is that given the hardware-specific constraints, intra-stage memory allocation matters more than inter-stage memory allocation. 

We argue that for a given stage, a configuration that uses more of the existing ALUs and more of the available memory should perform better than those with fewer ALUs and less memory. While there is no apparent downside to using all available memory, how memory is distributed across different registers matters. In fact, in the absence of practical reasons for favoring either small or large registers (\ie high variability in sizes) or assigning all registers the same size (\ie zero variability), we opt for a distribution of register-sizes that exhibits a moderate amount of variability. To this end and for simplicity, we consider here variability scenarios that result from allocating memory proportional to a register's order; that is, if $\{1, 2,..., A\}$ denotes the ordered set of registers in a stage, then the register sizes will be $\{S, 2S,..., AS\}$, where $S$ is a scaling factor.

\smartparagraph{Slice-n-Repeat (SnR) heuristic.}
Our goal is to develop a ``slice-n-repeat" algorithm that first finds a good register size configuration for a single stage (slice) and then ``repeats" that same configuration across all other stages. Based on the above observations, we focus here on synthesizing a slice configuration that meets the given sizing constraints, uses all the ALUs and all the memory available in a single stage, and satisfies the stipulated proportional relationship between allocated memory and registers' order.
As a result, our algorithm simply needs to solve a linear program for a single variable which is the scaling factor $S$ for register sizes. 
Our goal is to simply find the largest value of $S$, subject to the constraints that the maximum size allocated the the largest-sized register does not exceed the max memory limit for a single register and that the sum of the memory allocated to all the registers does not exceed the total memory limit for a stage. 

To illustrate, consider a setting where there are a total of eight ALUs (or registers) in a stage, where the maximum memory that we can allocate to a single register is 1~Mb, and the total available memory across the entire stage is capped at 2~Mb. In this case, the optimization problem has two constraints: $8S \leq 1$ and $\frac{8\times(8 + 1)}{2}S \leq 2$. Solving the resulting linear program yields the optimal solution $S=\frac{1}{18}$~Mb. Thus the register sizes are $\{\frac{1}{18}, \frac{2}{18}, \frac{3}{18}, \frac{4}{18}, \frac{5}{18}, \frac{6}{18}, \frac{7}{18}\frac{8}{18}\}$~Mb. 


Note that the SnR heuristic does not depend on any cost matrix as input and is, thus, insensitive to the input workload dynamics. It only depends on the set of input queries, which we assume in this paper to be given and fixed. 

%% file: runtime.tex
\section{Dynamic Operator Mappings}
\label{sec:runtime}

In this section, we describe how \system maps operators to data-plane registers at runtime (Section~\ref{ssec:mapping}) and how it deals with the fact that in practice, the input data and resulting cost matrix for a given window is not available at the start of the window when the operator mapping decisions have to be made (Section~\ref{sec:predict}). 



\subsection{Operator Mapping Problem}
\label{ssec:mapping}

\smartparagraph{Definition.}
Assuming that the current window ($w$) executes the initial query plan,
then the input for our operator-to-register mapping problem consists of (1)~the register-size configuration ($R$) used for $w$, and (2)~the cost matrix for the next window ($w'$): \ie the memory required ($B_o$) and number of input and output tuples ($N_{in}$ and $N_{out}$) for dataflow operators $o \in O$.  Here, $O$ is the set of all stateful operators. Our objective is to find a solution to this problem in the form of a configuration that specifies for window $w'$ the mapping of stateful operators $o \in O$ to the set of registers $r \in R$ that minimizes the load at the stream processor for $w'$.


\smartparagraph{Hardware and ordering constraints.}
Solving this optimization problem is made more complicated by domain-specific requirements that further restrict the set of feasible solutions. For example, we can have a one-to-many relationship between dataflow operators ($O$) and data-plane registers ($R$); that is, we can map an operator to multiple registers, but a register can only support one stateful operator at a time. Additionally, our problem's solution has to satisfy the ordering constraints that can arise from query partitioning. More concretely, if $o_i$ $\rightarrow$ $o_j$ are two stateful operators for a query, and $stages(o_i)$ represents the set of stages to which $o_i$ is mapped, then $max(stages(o_i)) < min(stages(o_j))$. For example, in the case of Query~\ref{superspreader}, we can only map the {\tt reduce} operator for stages greater than the ones assigned to the {\tt distinct} operator. 


\smartparagraph{NP-hardness.}
We prove that finding an optimal assignment of operators to registers is NP-hard. Our proof is by reduction from the 3-PARTITION problem and due to limited space, we present the formal proof in Appendix~\ref{sec:np-hard}. We also note that NP-hardness holds even for data-planes with a single stage and when there are no dependencies between stateful operators.

\smartparagraph{Greedy heuristic.}
We now present a greedy heuristic for solving the operator mapping problem that both results in minimal workload at the stream processor and is computable in a few tens of milliseconds --- i.e., an efficient and practical solution. Intuitively, an approach that sorts all the operators based on a utility score and greedily maps them based on that score should suffice. However, such an approach does not work well due to the ordering constraints. Instead of making decisions for individual operators, the proposed heuristic makes decisions for pairs of operators $o_i$ $\rightarrow$ $o_j$.\footnote{For queries with more than two stateful operators, pairs of operators are simply replaced by tuples.} The algorithm sorts these operator pairs based on their \emph{assignment score}. We define this score as the ratio of workload reduction at the stream processor to the memory footprint when the operator pair is mapped to the data plane. It iterates over the data-plane stages in order, and greedily selects the operator pairs with the highest assignment score for each stage. This approach prioritizes selecting operator pairs that help reduce the stream processor's workload with smaller memory overhead in the data plane. We provide the pseudocode for this algorithm in Appendices~\ref{ssec:greedy_deep} and \ref{ssec:heuralg1}.

%% file: predict.tex
\subsection{Predicting Future Workload}
\label{sec:predict}
Our formulation of the operator mapping problem assumes oracular knowledge of future cost matrices to compute new operator mappings at runtime. However, this assumption does not hold in practice. This section describes how \system uses the cost values from the past to predict future memory requirements and how it uses these predicted values to compute new operator mappings. 

\smartparagraph{Off-the-shelf learning model.}
Our goal is to use the past windows' cost values to predict the future cost values, one window at a time. Given the high cost associated with generating cost matrices, we consider learning models that do not require a large training dataset. Among various off-the-shelve time-series prediction tools, we chose the Winter-Holt's double exponential smoothing prediction (DESP) method~\cite{winters-forecasting}, mainly for its simplicity and ability to adequately capture the short-term trends in memory requirements for the different operators over time (we observe that for most operators, the median prediction error is less than $10~\%$).~\footnote{While there is merit in exploring the use of more advanced learning models for cost matrix prediction, we quantify in Section~\ref{sec:eval} below the maximum gains that such efforts can achieve.}

\begin{figure}[t] 
\begin{minipage}{1\linewidth}
\includegraphics[width=\linewidth]{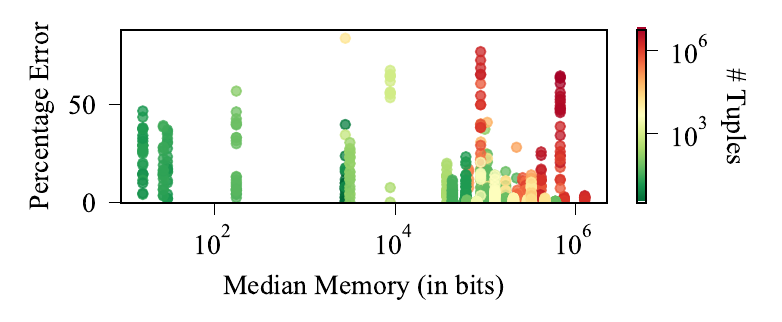}
\end{minipage}
\caption{Prediction error vs. operator's memory footprint.}
\label{fig:prediction_graphs}
\end{figure}

\smartparagraph{Scaling predicted values.}
Figure~\ref{fig:prediction_graphs} shows the relationship between prediction error and memory requirements. Here, each colored dot represents an operator, where the value on the x-axis indicates the (median) required memory, the y-axis indicates the prediction error, and the color encodes the operator's contribution to the load at the stream processor. The figure shows that even small prediction errors for operators with a high memory footprint can have a significant impact on the stream processor's load. We observed similar trends (not shown here) for the number of input tuples ($N$) and the variability of the memory required. Even with more sophisticated learning models, we expect non-zero prediction errors. Thus, to reduce the impact of prediction errors, we opted for an approach that scales the predicted values based on their effect on the stream processor's load. More specifically, we used the training dataset to cluster data-flow operators based on three features: number of input tuples ($N$), memory required ($B$), and variability in required memory. Here, to deal with $N$ and $B$ values across multiple windows, we considered their $95^{th}$ percentile values. This clustering effort produced 10 well-defined groups of operators, and we computed a different scaling factor for each group using Bayesian Optimization (BO)~\cite{bayesian}.

%% file: implementation.tex
\begin{table}[t]
\begin{footnotesize}
\begin{center}
\resizebox{\linewidth}{!}{%
\begin{tabular} {|c|l|cl|}
\hline
{\bf \#} & {\bf Query } & 
\multicolumn{1}{l}{\bf $|O|$} &
 \multicolumn{1}{l|}{\bf $B$~(bits)}\\
\hline
1 & Newly opened TCP Conns.~\cite{netqre}  & 1 & [2.8 M]\\
2 & SSH Brute Force~\cite{ssh-brute}  & 2 &[2.9~K, 85~K]\\ 
3 & Superspreader~\cite{opensketch}  & 2 &[1.2 M, 18.6 M]\\
4 & Port Scan~\cite{fast-portscan}  & 2 &[1 M, 18.6 M]\\
5 & DDoS~\cite{opensketch}  & 2 & [298.9 K, 4.4 M]\\
6 & TCP SYN Flood~\cite{netqre} & 3 & [[2.8 M], [8.7 M], [6 M]]\\
7 & TCP Incomplete Flows~\cite{netqre} & 2 & [[2.8 M], [928]]\\
8 & Slowloris Attacks~\cite{netqre} & 3 & [[1.4 M, 9 M], [18.6 M]]\\
\hline
\end{tabular}%
}
\end{center}
\end{footnotesize}
\caption{Telemetry queries.}
\label{tab:queries} 
\ifx \compress \undefined
\else
\vspace*{-1\baselineskip}
\fi
\end{table}

\begin{table*}[t]
\begin{footnotesize}
\begin{center}
\resizebox{\linewidth}{!}{%
\begin{tabular}{|l| c c | c |p{0.65\textwidth}|}
\hline
\textbf{} & 
\textbf{\begin{tabular}[c]{@{}c@{}}Iterative\\ Refinement\end{tabular}} &
\textbf{\begin{tabular}[c]{@{}c@{}}Query\\ Planning\end{tabular}} &
\textbf{\begin{tabular}[c]{@{}c@{}}Stream Processor\\ Workload\end{tabular}} 
& \textbf{Description} \\
\hline
MAX-DP & \xmark & \multirow{2}{*}{Static}   & 58.2~M & Executes as many dataflow operations as possible on the switch, e.g., Marple~\cite{marple}, UnivMon~\cite{univmon}, etc.\\ 
Sonata  & \cmark &  & 69.1~M & Computes refinement plan, register sizes, and operator mappings at compile time~\cite{sonata}.\\ 
\hline
\hline
MAX-DP-D & \xmark & \multirow{3}{*}{Dynamic} & 42.2~M & Uses SnR algorithm for register sizes and greedy heuristic for dynamic operator mappings.\\ 
\system-Oracle & \cmark &  & 77.2~K & Uses TOM and SnR algorithms for bootstrapping, and greedy heuristics for operator mappings.\\ 
\textbf{\system-Pred} & \cmark &  & \textbf{89.4~K} & Uses \system-Oracle initial query plan and predicted cost values to update operator mappings.\\ 
\hline
\hline
Optimal-MAX-DP & \xmark &  \multirow{2}{*}{Impractical}  & 41.8~M & Computes optimal query plans using MAX-DP's query planner for each window. \\ 
Optimal-Sonata & \cmark &   & 29.6~K & Computes optimal query plans using Sonata's query planner for each window. \\ 
\hline
\end{tabular}
}
\end{center}
\end{footnotesize}
\caption{Query-planning techniques emulated for evaluation.}
\label{tab:query_plans} 
\end{table*}

\section{Implementation}
\label{sec:implement}

\system augments Sonata's query planner, runtime, and data-plane driver modules. The most challenging part of implementing \system is to develop a new data-plane driver for Tofino-based targets, allowing them to map stateful operators to registers at runtime dynamically. This section describes how we enable dynamic operator mapping for a Tofino-model switch, which offers register-accurate chip simulation for Tofino in software~\cite{p4studio}, and report \system's runtime overheads. 




\smartparagraph{Enabling dynamic operator mapping at runtime.}
This task entails updating (1)~programs for stateful ALUs (\eg~ {\tt reduce} to {\tt distinct}), (2)~a set of keys for stateful operations (\eg ~{\tt sIP} to {\tt (sIP, dIP)}), and (3)~operator-to-register mappings. To enable these operations, \system encodes the decisions in the packets' metadata and uses match-action tables to dynamically update these fields at the end of each window. More specifically, it creates two categories of metadata fields, namely persistent fields, and ephemeral fields. Unlike persistent fields, \system reuses the ephemeral ones after every stateful operation. For example, the metadata fields carrying the information about which ALU programs (\eg {\tt distinct} or {\tt reduce}) to use are reusable across different stages and are thus ephemeral. 
We provide a more detailed description of our Tofino-based implementation and the metadata scheme in Appendix~\ref{sec:tofino}. A key lesson from this proposed approach is that compared to Sonata, \ system incurs additional metadata overhead of around 120 bits per ALU. Though this overhead is high, it should be possible to bring it down with a more efficient data-plane implementation.

\smartparagraph{Runtime overheads.}
\system performs the following operations at runtime:
(1)~read the register values from the data-plane targets, (2)~predict the input data (i.e., cost matrix values), (3)~run the greedy heuristic to compute operator mappings, and (4)~update the match-action tables in the data plane. Together, the first three operations take between 100-200~ms to complete. \system minimizes the impact of these delays by only using the cost matrix values from previous windows to predict the cost matrix for the next window. We verified that using no (or only partial) information about the cost matrix values in the current window for predicting the cost matrix has negligible impact on \system's performance. Our current implementation updates the match-action table entries sequentially, which incurs an overhead of 200-300~ms. Updating the match-action table entries in parallel can reduce the runtime overhead for this last operation to less than 10~ms in the future.

%% file: evaluation.tex
\section{Evaluation}
\label{sec:eval}
Our multi-faceted evaluation of \system described in this section results in the following key takeaways:
\begin{asparaitem}
    \item \system achieves near-optimal workload reduction. Compared to existing telemetry systems that use static query plans with (\eg Sonata) or without (\eg Marple) iterative refinement, \system reduces the stream processor's workload by more than \textbf{two orders of magnitude} (see Table~\ref{tab:query_plans} and Figure~\ref{fig:performance}). 
    \item \system's performance gains are robust to differences in \textit{switch constraints}, \textit{input workloads}, \textit{query workloads}, and \textit{training intervals} (see Figure~\ref{fig:dp_constraints_effect}).
    \item  Each of the bootstrapping and runtime algorithms described in Sections~\ref{sec:bootstrap} and ~\ref{sec:runtime} contribute to the achieved performance gains of \system and we quantify their individual contributions (due to limited space, see Appendix~\ref{sec:variants} for details).
    \item The ability of \system to dynamically map operators at runtime is demonstrated with a case study that leverages a Tofino-model switch~\cite{p4studio} (due to limited space, this use case is detailed in Appendix~\ref{ssec:tofino-case-study}). We provide the instructions to reproduce this case study over GitHub~\cite{tofino-github}.
\end{asparaitem}

\subsection{Setup}
\label{ssec:setup}

\begin{figure*}[t] 
\begin{minipage}{1\linewidth}
\begin{subfigure}[b]{.33\linewidth}
\includegraphics[width=\linewidth]{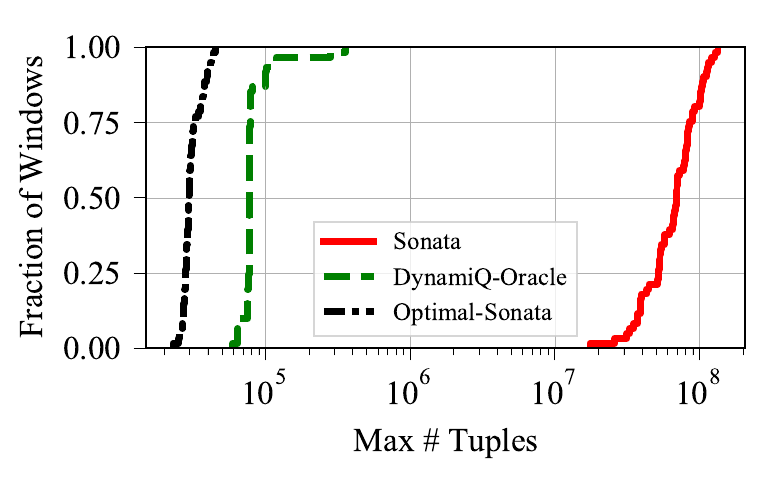}
\caption{\system vs. Sonata}
\label{fig:dynamiq-vs-Sonata}
\end{subfigure}
\begin{subfigure}[b]{.33\linewidth}
\includegraphics[width=\linewidth]{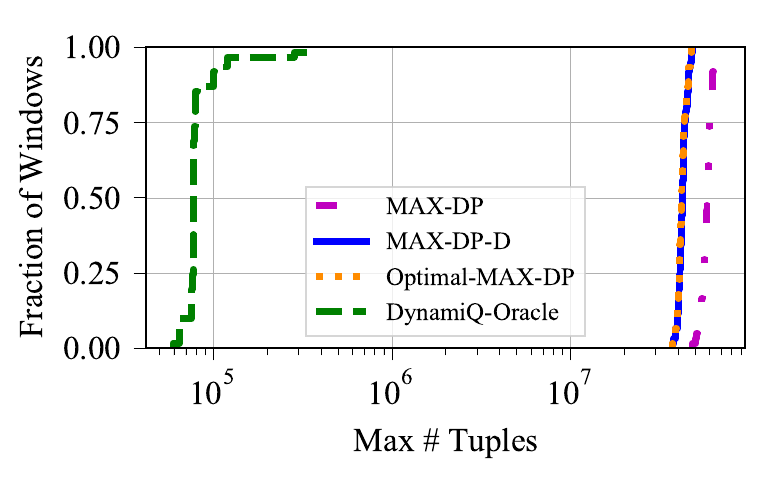}
\caption{\system vs. MAX-DP}
\label{fig:dynamiq-vs-maxdp}
\end{subfigure}
\begin{subfigure}[b]{.33\linewidth}
\includegraphics[width=\linewidth]{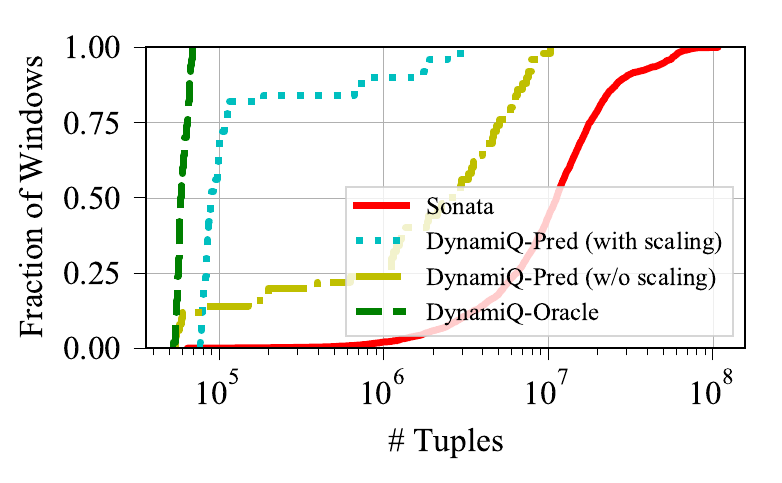}
\caption{\system-Oracle vs. \system-Pred}
\label{fig:prediction}
\end{subfigure}

\end{minipage}
\caption{Workload reduction at stream processor.
}
\label{fig:performance}
\end{figure*}

\smartparagraph{Telemetry queries.}
For a fair comparison, we used the same query workload as in Sonata~\cite{sonata}; that is, we used eight different telemetry queries that process layer 3 and 4 header fields for evaluation.\footnote{Out of 11 queries considered in~\cite{sonata}, three required payload information.} Table~\ref{tab:queries} lists the queries and shows how they differ in terms of the number of stateful operators ($|O|$) and their memory footprint ($B$). For queries with {\tt join} operators, we show the memory footprint for all stateful operators in individual sub-queries.



\smartparagraph{Packet Traces.}
We used four one-hour-long packet traces, collected three months apart within a single year and from the same source~\cite{caida15}. Similar to Sonata, we apply a speedup factor of $20\times$ to simulate a $100$~Gbps workload. Applying a window size of three seconds thus yields a total of four three-minute-long traces (each consisting of $60$ windows) that we used for evaluation. We split the resulting 60 windows into a test set and training set and use the training data to bootstrap initial query plans. Unless specified otherwise, we present the results for one of the four packet traces and only use one window for bootstrapping. We present results for the other three traces and longer training intervals in Section~\ref{ssec:sensitivity}.

\smartparagraph{Impact of packet orderings.}
We now describe how we estimate additional load at the stream processor when the required operator memory exceed the allocated memory in the data plane. Our approach (denoted as "average case") assumes a random ordering of the incoming keys, irrespective of the number of tuples that map to these keys. We now quantify the impact of packet ordering on the estimation of load. More specifically, we consider two extreme cases with skewed orderings. The first case is where only the keys with just one tuple arrive last and thus don't find space in the data plane. The second case we consider is where all the single-tuple keys arrive first and use all the allocated memory in the data plane. We refer to the first case as the "best case", and the second case as the "worst case". We describe in detail the calculations used for computing the load at the stream processor in Appendix~\ref{sec:load_est}. We also present the comparison between the three cases in Figure~\ref{fig:stale_cost_all_cases} in the appendix.

\smartparagraph{Impact of hash collisions.}
Hash collisions are intrinsic to stateful operations in the data plane. Sonata~\cite{sonata} evaluates the impact of hash collisions on load at stream processor and shows that to mitigate the impact of hash collisions, hybrid telemetry systems need to either over-provision the registers in the data plane or use more ALUs/stages to compile a stateful operators across multiple stages. In this paper, we don't evaluate the impact of hash collision as it affects both the static and dynamic query plans equally. 


\smartparagraph{Comparisons with existing systems.}
We compare the performance of \system with the performance of two other telemetry systems: MAX-DP and Sonata. We borrowed MAX-DP from ~\cite{sonata} because it emulates the set of telemetry systems, such as Marple~\cite{marple} and UnivMon~\cite{univmon}, etc., that don't use iterative refinement. Its query plan consists of offloading as many stateful operators as possible to the switch. We use Sonata because it is a system that does use iterative refinement. 
Table~\ref{tab:query_plans} provides descriptions for how we emulated these systems as well as for different variants of \system itself. Finally, we compare against Optimal-MAX-DP and Optimal-Sonata, which mimic the optimal (yet) impractical query planning techniques with and without iterative refinement, respectively. Also, note that except for \textbf{\system-Pred}, all other dynamic query-planning techniques in Table~\ref{tab:query_plans} assume oracular knowledge (\ie future cost matrix values) and are therefore not realizable in practice.

Quantifying the performance of the different query plans listed in Table~\ref{tab:query_plans} requires estimating the additional load at the stream processor when the required operator memory exceeds the allocated data-plane memory. In Appendix~\ref{sec:load_est}, we describe how we accurately estimate this metric without requiring packet-level simulations.

\smartparagraph{Targets.}
Query planning is most challenging when the gap between required and available data-plane resources is small. To emulate such resource-constrained settings, unless specified otherwise, we present results for a simulated PISA switch with 12 physical stages, eight stateful ALUs per stage, and $1.5$ Mb of register memory per stage. For each stage, a single register can use up to $0.75$ Mb of memory. We quantify the effect of switch resources on \system's performance in Section~\ref{ssec:sensitivity}.

\subsection{Workload at Stream Processor}
\label{ssec:perf}

To quantify the load at the stream processor for the different query-planning techniques listed in Table~\ref{tab:query_plans}, we perform the following experiment. For each technique, we consider a window and use the corresponding input cost matrix to compute the initial query plan according to the considered technique. We then apply this plan to all remaining 59 windows, and compute the max load observed at the stream processor. We repeat this process for each of the 60 windows and compute the distribution of these max load values. Table~\ref{tab:query_plans} (third column) reports these distributions' median values. 

\smartparagraph{\system vs. Sonata.}
Figure~\ref{fig:dynamiq-vs-Sonata} compares the performance of \system-Oracle (green line) with Sonata (red line) and Optimal-Sonata (black line). We observe that for most input windows, the initial query plans computed by \system-Oracle reduce the load at the stream processor by more than two orders of magnitude compared to Sonata. This result reinforces our argument that is combining query-specific iterative refinement and dynamic query planning results in significant workload reductions at the stream processor. This result also shows a marginal difference in performance between \system-Oracle and Optimal-Sonata---demonstrating the efficacy of \system's bootstrapping and runtime algorithms.

\smartparagraph{\system vs. MAX-DP.}
Figure~\ref{fig:dynamiq-vs-maxdp} compares \system's performance with that of MAX-DP and MAX-DP-D. We observe that in the absence of query-specific iterative refinement, the impact of dynamic query planning is minimal. We also note that the performance of MAX-DP-D, which employs both \system's register sizing heuristic and \system's greedy heuristic for dynamic operator mappings, is close to Optimal-MAX-DP, demonstrating the generality of \system's bootstrapping and runtime algorithms.

\smartparagraph{\system-Oracle vs. \system-Pred.}
\system-Oracle assumes oracular knowledge; that is, it knows the input data for future windows. Clearly, this assumption does not hold in practice, and Figure~\ref{fig:prediction} assesses the cost of using predicted cost values as input to compute the operator mappings at runtime. It shows that although \system-Pred sends more tuples to the stream processor compared to \system-Oracle, it still reduces the median load at the stream processor by more than two orders of magnitude compared to Sonata. The (right) tail behavior of the distribution for \system-Pred results from the high prediction errors for operators with a high memory footprint. This result also shows that our design choice to scale the predicted cost values helps reduce the workload at the stream processor. Compared to \system-Pred (with scaling), not scaling the predicted values (\ie \system-Pred (w/o scaling)) increases the load at stream processor by at least one order of magnitude. 
Note that unlike the experiments described above, to preserve temporal dependencies in the packet traces for the prediction model, we used the first ten windows for training \system-Pred's learning model and show the workload distribution for the next 50 windows. Table~\ref{tab:query_plans} reports the median workload over these 50 windows. 


\begin{figure}[t] 
\begin{minipage}{1\linewidth}
\includegraphics[width=\linewidth]{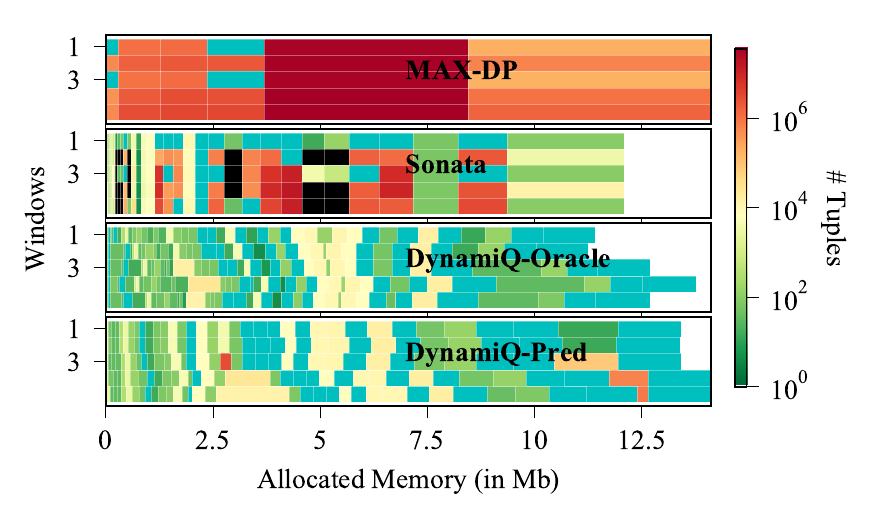}
\end{minipage}
\caption{Memory allocations (block size) and contributions to load at stream processor (block color) for MAX-DP, Sonata, \system-Oracle, and \system-Pred.}
\label{fig:viz}
\end{figure}

\smartparagraph{Visualizing memory allocation and load contribution.}
Figure~\ref{fig:viz} extends the visual presented in Section~\ref{sec:background} (Figure~\ref{fig:qp_dynamics}) and compares the overall performance of MAX-DP, Sonata,  \system-Oracle, and \system-Pred, respectively. For each of these query-planning techniques, the figure depicts five consecutive windows (rows) and shows for each window (row) the memory allocations (block size) and contributions to the load at the stream processor (block color) for the different operators. Here, the y-axis indicates the time (window). Sonata with its static query planning uses the first window's cost matrix as input to compute the query plan and applies that very plan for the future windows. 
We observe that MAX-DP uses a few large-sized blocks, and most of them contribute significantly to the load at the stream processor (\ie red/orange colors). Compared to MAX-DP, Sonata uses more (relatively) smaller-sized blocks whose contributions to the stream processor's load start to increase right after the first window (\ie yellow/red/black colors). 
In contrast,  the predominantly green/yellow/cyan-colored plots for \system-Oracle and \system-Pred for all five windows is visual evidence that the contributions of the individual operators of these two dynamic query planning techniques to the stream processor's load are significantly lower compared to Sonata.


Also note that compared to \system-Oracle, \system-Pred has a few operators in a couple of windows whose contributions to the stream processor's load are high (red/orange blocks). Their occurrence is due to large errors in predicting the memory requirements for these operators. We also observe that \system-Pred uses more data-plane memory than \system-Oracle and we attribute this observation to our design choice of scaling the predicted values (\ie overestimating the memory requirements).

\begin{figure}[t] 
\begin{minipage}{1\linewidth}
\begin{subfigure}[b]{.49\linewidth}
\includegraphics[width=\linewidth]{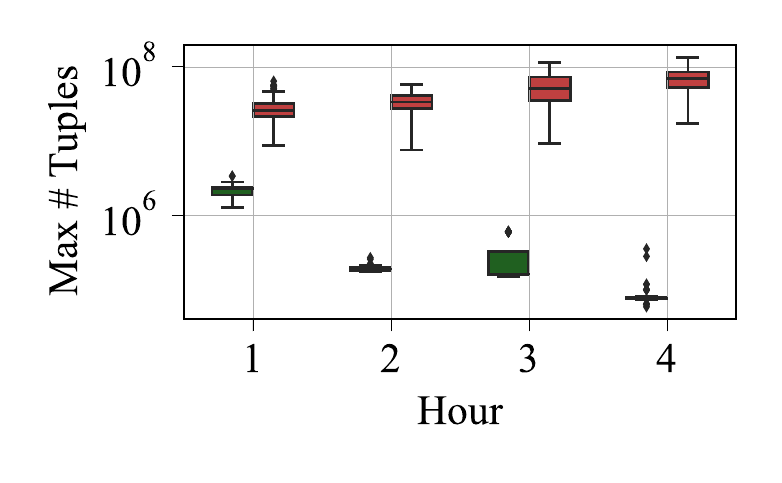}
\caption{Workloads}
\label{fig:workload_effect}
\end{subfigure}
\begin{subfigure}[b]{.49\linewidth}
\includegraphics[width=\linewidth]{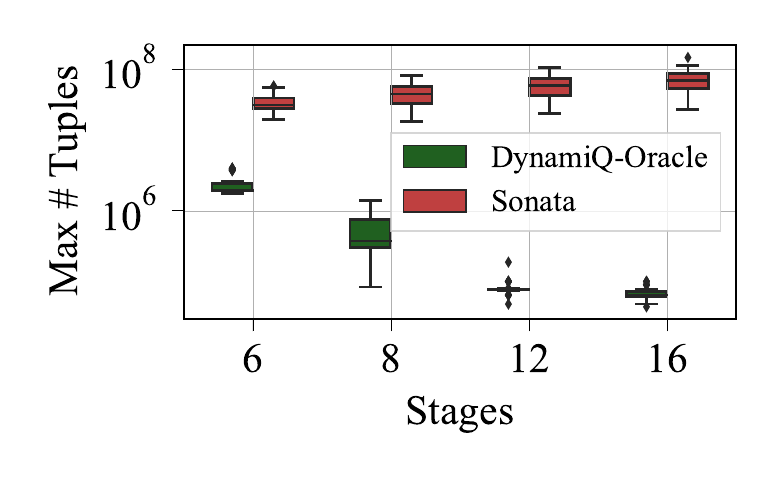}
\caption{Number of stages}
\label{fig:dp_targets_stages}
\end{subfigure}
\begin{subfigure}[b]{.49\linewidth}
\includegraphics[width=\linewidth]{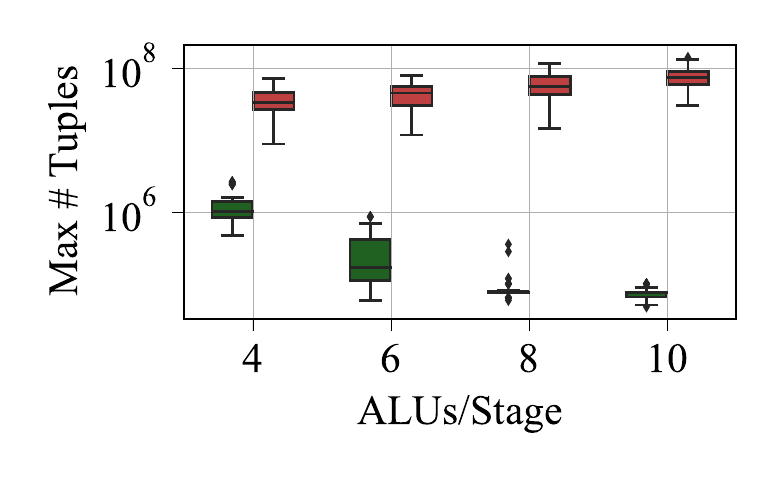}
\caption{Number of ALUs per stage}
\label{fig:dp_targets_alus}
\end{subfigure}
\begin{subfigure}[b]{.49\linewidth}
\includegraphics[width=\linewidth]{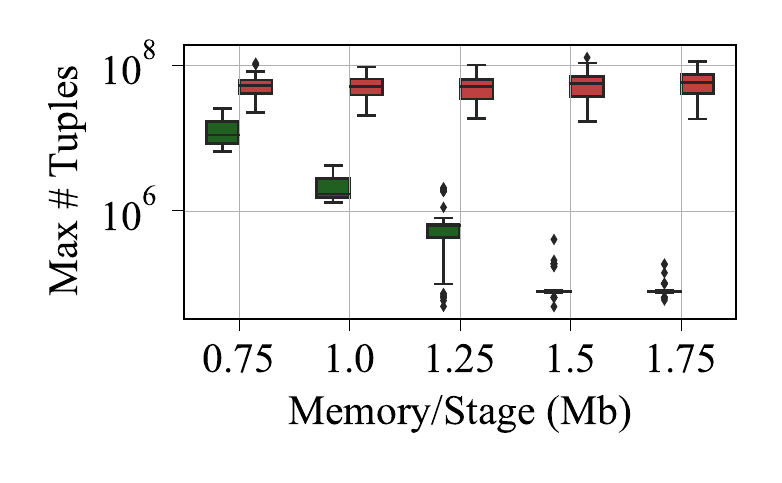}
\caption{Memory per stage}
\label{fig:dp_targets_mem}
\end{subfigure}
\begin{subfigure}[b]{.49\linewidth}
\includegraphics[width=\linewidth]{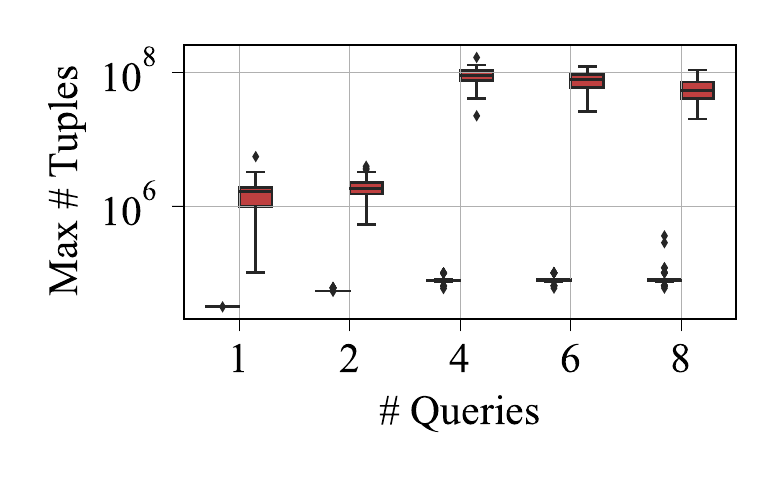}
\caption{Query workload}
\label{fig:n_queries}
\end{subfigure}
\begin{subfigure}[b]{.49\linewidth}
\includegraphics[width=\linewidth]{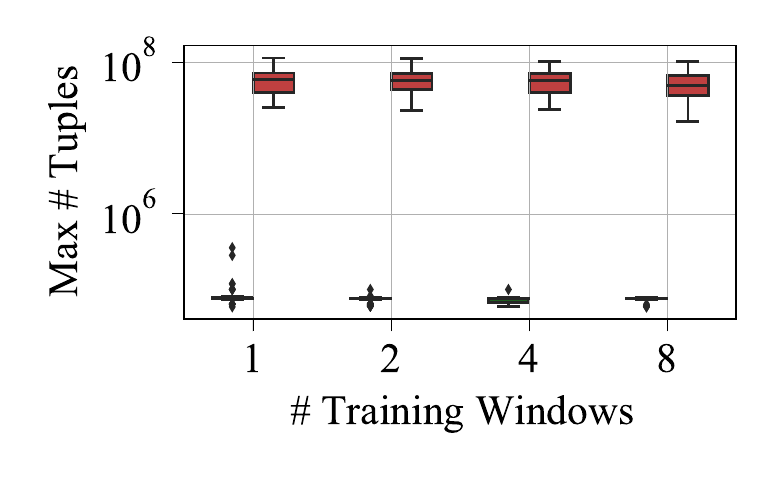}
\caption{Training windows}
\label{fig:training_wins}
\end{subfigure}
\end{minipage}

\caption{Sensitivity analysis: \system's (and Sonata's) performance for different workloads in (a), data-plane targets in (b), (c), and (d), query workloads in (e), and training windows in (f). 
}
\label{fig:dp_constraints_effect}
\end{figure}

\subsection{Sensitivity Analysis}
\label{ssec:sensitivity}
We next demonstrate \system's robustness to various input workloads, switch constraints, query workloads, and training intervals.

\smartparagraph{Input traffic workload.}
To assess the (in)sensitivity of \system's performance to different input workloads, Figure~\ref{fig:workload_effect} shows the workloads at the stream processor for four different hours of packet trace data. We observe that the overall trend (\ie \system performs better than Sonata and Optimal better than \system) is consistent across the different datasets. Compared to Sonata, we observe smaller gains for \system for the first hour which can be explained by the fact that the variability in total memory is high in the first hour, which in turn affords the runtime fewer opportunities for finding good dynamic operator mappings.




\smartparagraph{Switch constraints.}
To assess \system's performance for different data-plane target configurations, we first quantify the impact of the number of stages and number of ALUs, keeping the total memory constant. As can be seen from Figure~\ref{fig:dp_targets_stages} and Figure~\ref{fig:dp_targets_alus}, decreasing (increasing) the number of stages or ALUs makes dynamic query planning less (more) effective in terms of the max number of tuples that are sent to the stream processor. Second, with respect to quantifying the impact of total memory for a fixed number of stages and ALUs,  Figure~\ref{fig:dp_targets_mem} shows that the difference in performance between \system and Sonata diminishes as the total available memory decreases. This observation reflects the fact that there are less opportunities for the runtime to redistribute memory when total memory is already tight (\ie operating in resource-constrained settings).

\smartparagraph{Query workload.}
Figure~\ref{fig:n_queries} compares \system's and Sonata's performance for different query workloads. We observe that \system consistently outperforms Sonata, and the performance gap is smaller for workloads with fewer input queries.

\smartparagraph{Training intervals.}
With multiple windows of training data (cost matrices), Sonata uses the median value as an input to its ILP solver. Figure~\ref{fig:training_wins} compares \system's and Sonata's performance for different training intervals. We observe that the number of training intervals have minimal impact on the difference in performance between the two.

%% file: related-new.tex
\section{Related Work}
\label{sec:related}


\smartparagraph{Network streaming analytics systems.} 
Modern data plane technologies have had a profound impact on network telemetry~\cite{in-band,minlan-topdown,in-network-on-demand}. Existing network telemetry systems can be roughly grouped into purely host-based platforms~\cite{gigascope, chimera, opensoc, netsight, netqre, dshark, trumpet, switchpointer, pathdump}) and data plane-only systems~\cite{beaucoup, univmon, everflow, opensketch, nitrosketch}), and include as of late also hybrid designs~\cite{sonata, newton, marple, starflow, flowradar, concerto, carpe, herd, turboflow}. 
Most existing hybrid telemetry systems focus on mechanisms for compiling a subset of dataflow operators in the switch (\eg Marple~\cite{marple}, Newton~\cite{newton}, etc.). Among hybrid telemetry systems, only Sonata~\cite{sonata} is concerned with algorithms for computing static query plans that make effective use of limited switch resources for a given query workload, traffic workload and data-plane target.
\system is a hybrid system for single-switch settings and static query workloads that is capable of dynamically updating query plans to successfully respond to input workload dynamics.


\smartparagraph{Database systems.} 
A large body of literature in the database research community deals with the problem of query optimization in general (\eg see~\cite{hellerstein2017} and references therein) and adaptive (or dynamic) query processing in particular (\eg see~\cite{adaptive-qp} and references therein). In the context of traditional (single-machine) database systems, dynamic query planning concerns runtime feedback to modify intra-query processing in ways that provide more efficient resource utilization. \system has much in common with recent work such as Qoop~\cite{qoop} that focuses on big data analytics tasks in multi-node clusters and makes a case for ``on-the-fly" query (re)planning; that is, switching query plans during query execution to adapt to changing resource availability. \system's dynamic query planning is an instance of such on-the-fly query (re-)planning that adapts the query plans to make the best use of limited data-plane resources. 

%% file: conclusion.tex
\section{Limitations \& Future Directions}
\label{sec:future}
An ideal network streaming analytics system should be able to compute effective query plans for executing \textit{dynamic query workloads} in \textit{distributed settings} for \textit{dynamic input workloads}. While Sonata has been designed and works well for static input and query workloads in single-switch settings, \system represents an important improvement over Sonata in the sense that it considers dynamic query planning as a means to effectively deal with dynamic input workloads in single-switch settings. Moreover, \system's ability to dynamically map operators to data-plane registers at runtime lays the foundation for system designs that can readily accommodate dynamic query workloads; that is, compute effective query plans for query workloads where the set of operators that requires mapping changes as the query workload changes. 

While much work remains, we view \system as a promising step towards achieving the stated ``ideal". For example, although a few of the existing network telemetry systems (\eg Marple~\cite{marple}, Newton~\cite{newton}) operate in distributed settings, it is unclear if they can compute effective query plans for dynamic input/query workloads and/or arbitrary topologies/routing. We plan to extend \system's query-planning setup for distributed settings as part of our  future work. At the same time, realizing that \system provides excellent opportunities to leverage machine learning (\eg for workload prediction) to reduce the load at the stream processor, we also plan to explore if or how leveraging machine learning can improve query planning for network telemetry systems more generally.

\section{Conclusion}
\label{sec:conclusion}
For network streaming analytics systems to leverage programmable data-plane targets effectively requires dynamic query planning. Our new prototype \system satisfies this requirement by first computing effective initial query plans and then dynamically changing them by updating the operator mappings at runtime. We show that \system achieves near-optimal workload reduction and can reduce the load at the stream processor by more than two orders of magnitude compared to existing systems, such as Sonata, Marple, etc., that use static query plans. These performance gains are also robust to differences in switch constraints, input traffic workloads, query workloads, etc. We also demonstrated how we integrate \system's runtime algorithms with Tofino-based data-plane targets.

%% file: appendix.tex
\newpage
\appendix

\section{Static Query Planning}
\label{sec:alt_static}
In this part of the appendix, we report on the different techniques we investigated to characterize the limitations of static query planning. We divide this exploration into two broad categories: (1)~\textit{overprovisioning}, where we scale the input for query planning to handle future input workload dynamics, and (2)~\textit{overfitting}, where we extend Sonata's query planner to find the optimal query plan for all the input windows.

\subsection{Over-provisioning.}
The efficacy of static query plans depends critically on how well they can account for the dynamic changes inherent in future windows over longer time horizons. The inability to effectively handle these future changes results in ineffective query plans, especially for scenarios where network resources are scarce. To handle future workload dynamics, one can simply scale the cost matrix and then compute query plans to address the scaled memory requirements. Such an approach can reduce the number of under-provisioning cases which in turn can decrease the load at the stream processor. However, such over-provisioning comes at a cost -- a significant fraction of available switch memory may remain unused. We now discuss two different approaches for scaling the cost matrices. 

\smartparagraph{Uniform scaling.}
One simple approach to account for the unknown changes encountered in future windows is to uniformly over-provision the input cost matrix by, for instance, doubling the memory requirements for all the cost matrix elements and then using these over-provisioned cost matrices as input to Sonata's ILP solver. While the resulting static query plan will work well for operators whose memory requirements for future windows are higher than their median value, for the remaining operators, this approach results in wasteful allocation of switch memory that is already in short supply to start with and leads to increased loads at the stream processor, especially for windows that see plenty of wasted switch memory.

\smartparagraph{Selective scaling.} 
Different cost matrix elements have different variability. Also, the impact of over-provisioning for different operators (cost matrix elements) on the stream processor's load is different. Thus, instead of using a uniform scaling factor for the entire cost matrix, one can use different scaling factors for different cost matrix elements. We explored a close-loop technique to find the right over-provisioning levels for different operators. In particular, we used Bayesian Optimization~\cite{bayesian} to compute the over-provisioning levels for different cost matrix elements. To reduce the search space, we only computed query-specific over-provisioning levels. Here, we optimized a black-box function $F$ which takes as input the query-level scaling factor and gives as output a score indicating the workload at the stream processor across all input windows. The score that we used in our experiments minimizes the sum of log of workload at the stream processor in each window: $\sum_{i \in \text{windows}} \log_2 (\text{Load at SP})_i$.

In principle, while such an approach should help with finding good static query plans, it does not generalize well, especially with limited training data. Given the high cost of generating cost matrices, both for collecting the raw data and then processing it, we expect limited availability of training data in practice.

\subsection{Over-fitting.} 
To compute its query plan, Sonata uses the median value of all the cost matrices as an input to its ILP solver. One possible improvement over Sonata's current approach is to extend its query-planning ILP to optimize for all the input windows simultaneously. However, since the number of constraints for this extended ILP increases linearly with the number of input windows and since the computational complexity of computing an approximate solution of the ILP for a single window is already high, we expect this approach to take orders of magnitude more time to converge to a (near-)optimal solution. Moreover, the resulting query plan will be over-fitted for the input cost matrices and will therefore suffer from the same problems as Sonata's original query planner when it comes to handling the traffic workload dynamics encountered during future windows.


\begin{figure}[t!]
\begin{lstlisting}[language=query2,basicstyle=\footnotesize, 
basicstyle=\footnotesize, numbers=left,xleftmargin=2em,frame=single,
framexleftmargin=2.0em, captionpos=b, label=ddos, caption=Detect DDoS Attacks.]
packetStream(W)
.filter(p => p.proto == 17)
.map(p => (p.dIP, p.sIP))
.distinct()
.map((sIP, dIP) => (dIP,1))
.reduce(keys=(dIP,), f=sum)
.filter((dIP, count) => count > Th)
\end{lstlisting}
\ifx \compress \undefined
\else
\vspace*{-1\baselineskip}
\fi
\end{figure}

\section{Accurate Query Evaluation}
\label{ssec:query_acc}

In this part of the appendix, we describe how we ensure accurate query evaluation. More specifically, we show that for queries with multiple stateful operators, if the first stateful operator requires more memory than provisioned in the data plane, we need to send the output of the first operator to the stream processor and bypass the remaining stateful operators in the data plane. We also show that changing refinement plans for a query at runtime affects its accuracy. Thus, \system computes the refinement plans using the bootstrapping algorithm at compile time itself.  

\smartparagraph{Handling under-provisioned data-plane registers.}\\
Under-provisioned data-plane targets can result in increase in hash collisions. If a packet's key generates a collision, we send it to the stream processor. We then perform the aggregation operation for all the packets belonging to this key at the stream processor and join it with the result of the ones aggregated in the data plane to report the final output. This partial execution approach works well if we observe hash collisions for the query's last operator that gets executed at the switch. However, if the collision occurs for an intermediate aggregation operation, sending the packet to the stream processor affects accuracy. For example, consider Query~\ref{ddos} to detect victims of DNS reflection attacks. If the query plan is such that we execute up to line 6 in the data plane, then handling hash collisions for the {\tt distinct} operator is non-trivial. More concretely, if we apply the remaining operators for the keys that observe collision for the {\tt distinct} operator and join it with the output of the {\tt reduce} operator for the other keys, the output will be different. To illustrate, consider the set of input tuples as $\{(a,b), (a,b), (a,c), (c,d)\}$, and threshold (line 7) of $2$. Due to collision, we send the tuple $(a,c)$ to the stream processor. Here, with partial query execution, the output is a null set, different from the actual output $(a,2)$.

To ensure accurate query evaluation, we have two options. The first option is to send the output of the {\tt distinct} operator in the case of a collision to the stream processor, bypassing the remaining operators in the data plane. The other option is to send the output of the {\tt distinct} operator back to the data plane. Given the complexity of realizing this second option, we chose to use the first option.

\smartparagraph{Changing refinement plans.} 
In Section~\ref{sec:background}, we mentioned that changing the refinement plan at runtime affects a query's accuracy. Below, we provide an explanation and then describe a technique that can address this issue but is not practical in the sense that it cannot be executed at runtime. However, the technique can be used in practice when the frequency of changing refinement plans is of the order of hours or days. In this case, the changes are intended to handle significant shifts in the overall traffic patterns. 

A straw man approach is to abruptly change the refinement plan from one plan to the other. To recall, iterative refinement works by filtering the input workload of an operator $A$ based on the previous window's output of another operator $B$. The operator $B$ executes at the previous refinement level relative to that of the operator $A$. Thus, if \system changes the refinement plan in a window, the filtering mechanism is severely impaired. Consider a concrete example -- refinement plan for window $w_1$ is $0 \rightarrow 8 \rightarrow 24 \rightarrow 32$ (call it $R_1$) and the refinement plan for the successive window $w_2$ is $0 \rightarrow 16 \rightarrow 28 \rightarrow 32$ (call it $R_2$). Note that the output from the operator executing at \texttt{\textbackslash 8} in window $w_1$ would be used for filtering the input workload for the operator executing at \texttt{\textbackslash 24} in window $w_2$. However, when abruptly changing from $R_1$ to $R_2$, there is no operator executing at \texttt{\textbackslash 24} in $w_2$. Moreover, the new operators introduced in $w_2$ (\eg \texttt{\textbackslash 28}) lack the information to apply the filter operation for their input workload. In essence, changing refinement plans at runtime requires making these abrupt changes to refinement plans, but making such changes affects the query's accuracy. 

To preserve query accuracy, we can iteratively change the refinement plan over multiple consecutive windows. In the example described above, instead of adding all the operators for refinement plan $R_2$ in $w_2$, we only add operator \texttt{\textbackslash 16} in $w_2$, \texttt{\textbackslash 28} in $w_3$, and so on. This approach never introduces any operators that lack the information to create an input workload filter in any window. Since this approach takes $k$ windows to modify the plan (here $k$ is the number of refinement levels), it is not suited for making changes at runtime. Another related problem concerns predicting the memory requirements for new operators. In particular, the proposed approach requires estimating the memory requirements for new operators at different granularities for which we have no historical data. Estimating the memory requirements of such operators will require a more sophisticated learning model that leverages the spatial relationship between operators at different levels of granularities. We leave this exploration for future work.

\section{Operator Mapping Deep Dive}

\input{greedy_heuristic}

\input{tofino}

\input{marple}

\section{Load at Stream Processor}
\label{sec:load_est}
This part of the appendix describes how, as part of our evaluation, we extended Sonata's simulator to emulate each of the query-planning techniques listed in Table~\ref{tab:query_plans}. Some of these techniques (\eg \system-Pred, Sonata, etc.) require us to estimate the additional load at the stream processor when the required operator memory exceeds the allocated memory in the data plane. While we can compute this load accurately using a packet-level simulation, such an approach is prohibitively slow.

Instead of packet-level simulation, Sonata uses the stream processor to aggregate cost values for each window and uses these values to estimate the stream processor's load. However, such aggregation loses all information about the order in which packets are processed in the data plane. 

\smartparagraph{How we estimate load at the stream processor?}
To address this issue, we approximate the additional load using the \textit{average} value. Specifically, we report the product of the average number of tuples-per-key and the number of keys that cannot fit in the data-plane register. For example, if the average number of tuples-per-key for an operator is 10 and given the register size, five keys cannot fit in, then our estimate of the additional load at the stream processor will be 50. 


Recall that for every operator, the cost matrix provides three metrics: $(N_{in}, N_{out}, B_{req})$. Here, $N_{in}$ denotes the number of input tuples for the operator, $N_{out}$ the number of output tuples after applying the stateful operation, and $B_{req}$ is the required operator memory. The memory allocated for this operator in the data plane is denoted by $B_{alloc}$. If $B_{alloc}$ is greater than $B_{req}$, the data-plane register is over-provisioned. In this case, the number of tuples sent to the stream processor ($N_{sp}$) is given by $N_{sp} = N_{out}$. For the cases where the data-plane register is under-provisioned (\ie  $B_{alloc} < B_{req}$), we estimate the load due to keys that find an entry in the data plane to be $N_{out} \frac{B_{alloc}}{B_{req}}$ and the additional load due to hash collision to be $N_{in} \frac{B_{req}-B_{alloc}}{B_{req}}$. Therefore, the total load due to under-provisioning is estimated as $N_{sp}^{\text{avg}} = N_{out} \frac{B_{alloc}}{B_{req}} + N_{in} \frac{B_{req}-B_{alloc}}{B_{req}}$.

We now describe the calculations used for estimating the best and the worst case load estimates. Recall that in the best case, the keys with only a single tuple arrive at last. We can express this load as $N_{sp}^{\text{best}} = N_{out} +  \frac{B_{req}-B_{alloc}}{B_{key}}$. Here, $B_{key}$ denotes the memory footprint of a single key. For the worst case, all the single-tuple keys arrive first and use all the allocated memory in the data plane. We can express this load as $N_{sp}^{\text{worst}} = N_{in} - \frac{B_{alloc}}{B_{key}}$. 

Figure~\ref{fig:stale_cost_all_cases} in the appendix shows the load at the stream processor for each of these three cases. The experiment is similar to the one we conducted for Figure~\ref{fig:sp-workload}. In particular, we used the first window's cost matrix to compute a static query plan using Sonata's query planner. We then apply this plan on all the 60 windows. We used the three different expressions described above to compute the load at the stream processor for average, best, and worst case, respectively.

\begin{figure}[t]
\begin{minipage}{1\linewidth}
\includegraphics[width=\linewidth]{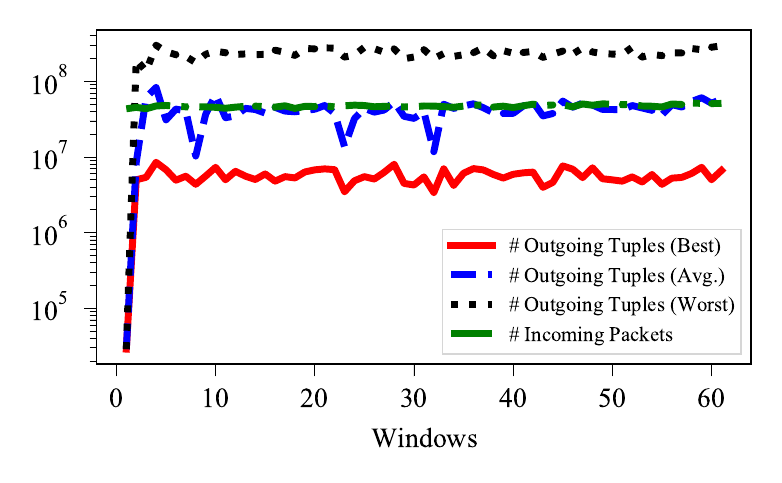}
\end{minipage}
\caption{Output workload because of a Sonata's static query plan compared to the input workload. The Output workload is estimated and the figure shows the best, worst and the average cases.}
\label{fig:stale_cost_all_cases}
\end{figure}

\begin{figure}[t] 
\begin{minipage}{1\linewidth}
\begin{subfigure}[b]{1\linewidth}
\includegraphics[width=\linewidth]{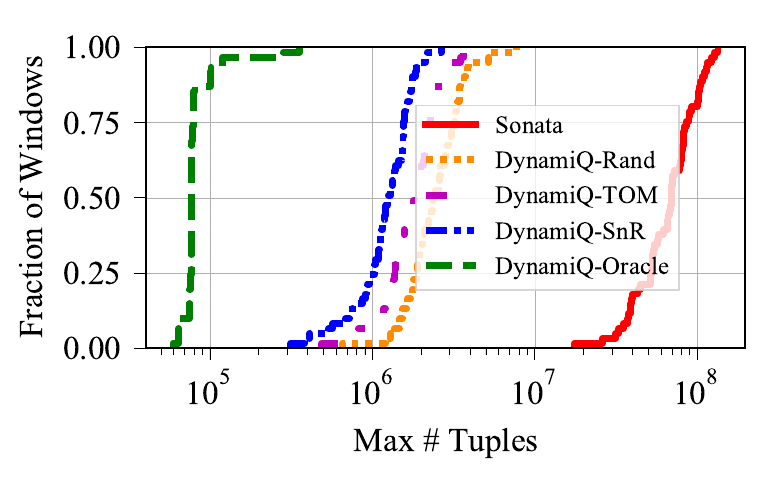}
\caption{Variants of Dynamic}
\label{fig:dynamiq_approaches}
\end{subfigure}
\begin{subfigure}[b]{1\linewidth}
\includegraphics[width=\linewidth]{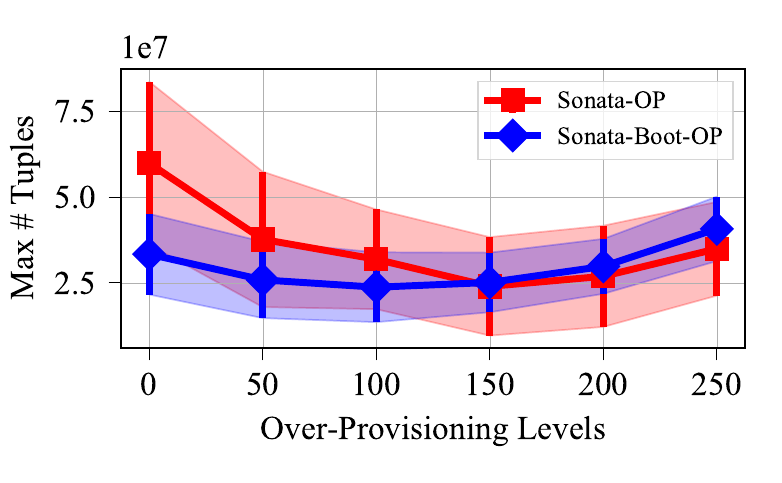}
\caption{Variants of Static}
\label{fig:op_approaches}
\end{subfigure}

\end{minipage}
\caption{Workload reduction at stream processor. 
}
\label{fig:static_query_plans}
\end{figure}

\begin{table*}[t]
\begin{footnotesize}
\begin{center}
\resizebox{\linewidth}{!}{%
\begin{tabular}{|l|c c  c | c |p{0.55\textwidth}|}
\hline
\textbf{} & \textbf{\begin{tabular}[c]{@{}c@{}}Refinement\\ Plan\end{tabular}} & \textbf{\begin{tabular}[c]{@{}c@{}}Register\\ Sizes\end{tabular}} & \textbf{\begin{tabular}[c]{@{}c@{}}Operator\\ Mapping\end{tabular}} & \textbf{\begin{tabular}[c]{@{}c@{}}Stream Processor\\ Workload\end{tabular}} 
& \textbf{Description} \\
\hline
Sonata-OP & Sonata & Sonata & \multirow{2}{*}{Static}  & 32.6~M & Uses scaled-up cost matrix values as input to Sonata's query planner.\\ 
Sonata-Boot-OP & TOM & SnR & & 23.8~M & Uses scaled-up cost matrix values only to compute static operator mappings. \\ 
\hline
\hline
\system-Rand & Sonata & Sonata & \multirow{3}{*}{Greedy} &  2.4~M & Uses Sonata's initial query plan, but dynamically updates operator mappings at runtime.\\ 
\system-SnR & Sonata & SnR &  & 1.3~M & Initial query plan is based on Sonata's refinement plan and SnR-based register sizes.\\ 
\system-TOM & TOM & Sonata & & 1.8~M & Initial query plan is based on TOM-based refinement plan and Sonata's register sizes.  \\ 

\hline

\end{tabular}
}
\end{center}
\end{footnotesize}
\caption{Query-planning techniques to emulate various static and dynamic variants. For Sonata-OP and Sonata-Boot-OP, we show the results for 100 \% overprovisioning in column five.}
\label{tab:query_plans_vars} 
\end{table*}

\section{Query Planning Variants}
\label{sec:variants}
In this part of the appendix, we evaluate the different variants of static and dynamic query planning techniques that are listed in Table~\ref{tab:query_plans_vars} and that we implemented as part of our evaluation. 
Both the static and dynamic variants differ in terms of their choice of algorithms for computing initial query plans (\ie refinement plans and initial register sizes). 

\smartparagraph{Variants of dynamic query planning.}
To quantify the contributions of initial query plans, we repeat the experiment described in Section~\ref{ssec:perf}, and Figure~\ref{fig:dynamiq_approaches} shows for each variant the distribution of the load at the stream processor. Here, the gap between \system-Rand and Sonata quantifies the contributions of dynamic operator mapping because both techniques use the same initial query plan (\ie the one that consists of the refinement plan and register size configuration computed by Sonata's query planner). Next, the gap between \system-Rand and \system-TOM, and \system-Rand and \system-SnR, quantifies the contributions of our proposed TOM-based and SnR algorithms, respectively.  In both cases, the gains are marginal and demonstrate that both SnR and TOM-based algorithms are not very useful by themselves. Significant workload reductions at the stream processor are only possible when used in combination (\ie \system-Oracle). 


\smartparagraph{Variants of static query planning.}
Figure~\ref{fig:op_approaches} shows the performance of different variants of static query planning. Compared to Sonata, Sonata-Boot reduces the load at the stream processor by half (zero percent over-provisioning case). As we increase the level of over-provisioning, the load initially decreases because over-provisioning has the effect of reducing under-provisioning for many operators. However, as the over-provisioning level keeps increasing, the amount of wasted memory also increases, and so does the resulting load at the stream processor. Compared to \system, the workload reductions attainable with over-provisioned static query plans are only marginal. 



\begin{figure}[t] 
\begin{minipage}{1\linewidth}
\begin{subfigure}[b]{.49\linewidth}
\includegraphics[width=\linewidth]{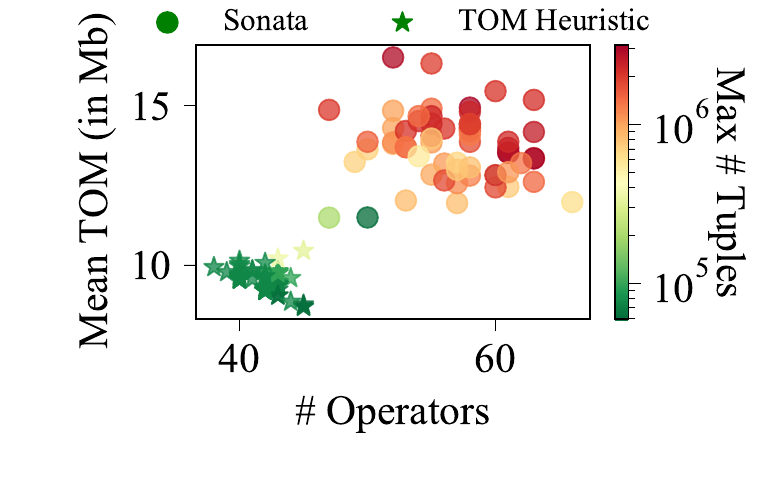}
\caption{TOM vs. Sonata}
\label{fig:tom_vs_sonata}
\end{subfigure}
\begin{subfigure}[b]{.49\linewidth}
\includegraphics[width=\linewidth]{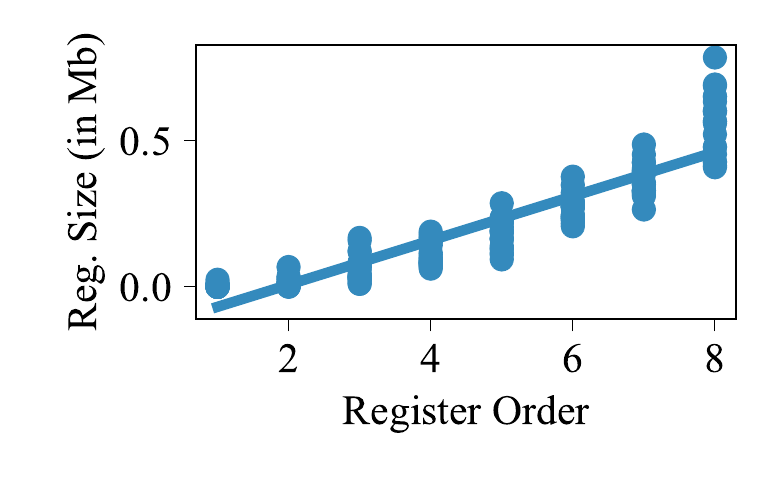}
\caption{Performant Sonata slices}
\label{fig:reg_slice_linear}
\end{subfigure}
\end{minipage}
\caption{Bootstrapping algorithms. Here, (a)~compares TOM-based refinement plans with Sonata; and (b)~shows the register size allocation for performant Sonata slices.}
\label{fig:bootstrapping_algos}
\end{figure}

\label{ssec:dynamiq}



\section{Bootstrapping Deep Dive}
This part of the appendix considers the two proposed bootstrapping algorithms in more detail. More specifically, we provide an intuition for why the two proposed heuristics are effective in providing opportunities for workload reduction with dynamic operator mappings at runtime. 

\smartparagraph{TOM-based vs. Sonata's refinement plans.}
Figure~\ref{fig:tom_vs_sonata} shows the relationship between mean-TOM (y-axis), the number of operators (x-axis), and the max load at the stream processor across all 60 windows (indicated by the data points' color) for refinement plans computed by TOM (star-shaped data points) and Sonata's query planner (circles), respectively. 
The plot clearly separates the refinement plans selected by our TOM-heuristic (\ie points in lower left corner) from those computed by Sonata (\ie points towards the upper right corner) and shows that the characteristic features of the TOM-based plans (\ie low mean TOM and small number of operators) enable them to handle the workload dynamics more efficiently than Sonata's plans, resulting in a significant reduction in load for the stream processor (\ie green- vs red-colored nodes). 

\smartparagraph{SnR-based vs. Sonata's register sizes.}
To better understand why the proposed SnR heuristic works well, we conducted the following experiment to be able to compare the performance of the register sizes computed by Sonata's query planner and their slice-n-repeat counterparts. For the experiment, we selected among the slice-n-repeat configurations the ones which result in minimum load at the stream processor. In Figure~\ref{fig:reg_slice_linear}, we observe that for these performant slices, (1)~the load they impose on the stream processors is similar to that of \system, (2)~the total memory and total number of ALUs used by these configurations are high, and (3)~the register sizes are approximately proportional to their order or IDs. These empirical observations explain why our SnR heuristic is effective---it tries to emulate the register size distribution of these performant slices.

\clearpage

%% file: greedy_heuristic.tex
This part of the appendix provides a formal proof that the problem of finding an optimal assignment of dataflow operators to data-plane registers is {\sf NP}-hard. The {\sf NP}-hardness holds even when there is just a single stage in the data plane, and there are no dependencies between the operators. For this reason we do not include these aspects in the formal definition of the problem that we show is {\sf NP}-hard. Nevertheless, the hardness of the {\sc Optimal Assignment} problem defined below directly translates to hardness of finding an optimal operator mapping. 

\subsection{Problem Formulation}
\label{sec:np-hard}



Let $R$ and $O$ be sets of registers and operators, respectively. Each register $r \in R$ is equipped with a non-negative integer {\em capacity} $c_r$. Each operator $o \in O$ is equipped with a positive integer {\em size} $s_o$, as well as two non-negative integer costs; $c_o^u$ is the {\em unsatisfied} cost, while $c_o^s$ is the {\em satisfied} cost of operator $o$. We require that $c_o^s \leq c_o^u$ for every operator $o$.

An {\em assignment} is a function $\alpha : R \rightarrow O$, and the {\em satisfaction ratio} of an operator $o$ with respect to an assignment $\alpha$ is denoted by $\rho(o, \alpha)$ and defined as 
$$\rho(o, \alpha) = \min\left(\frac{\sum_{r \in R ~:~ \alpha(r) = o} c_r}{s_o}, 1 \right)$$ 
An operator $o$ is said to be {\em unsatisfied} by the assignment $\alpha$ if $\rho(o, \alpha) = 0$, {\em satisfied} if $\rho(o, \alpha) = 1$, and {\em partially satisfied} otherwise. 

The cost of operator $o$ with respect to the assignment $\alpha$ is denoted by $c(o, \alpha)$ and defined as $c(o, \alpha) = c_o^s \cdot \rho(o, \alpha) + c_o^u \cdot (1 - \rho(o, \alpha))$. Thus the cost of a satisfied operator $o$ is the satisfied cost $c_o^s$, the cost of an unsatisfied operator is the unsatisfied cost $c_o^u$, while the cost of a partially satisfied operator is somewhere between $c_o^s$ and $c_o^u$ based on the satisfaction ratio $\rho(o, \alpha)$. The cost of an assignment $\alpha$ is denoted by $c(\alpha)$ and defined to be the sum $c(\alpha) = \sum_{o \in O} c(o, \alpha)$ of the costs of all operators $o$ in $O$ with respect to $\alpha$.

In the {\sc Optimal Assignment} problem, we are given as input a set $R$ of registers, a set $O$ of operators, a capacity $c_r$ for every register, and a size $s_o$, a satisfied cost $c_o^s$ and an unsatisfied cost $c_o^u$ for every operator $o$. The task is to compute an assignment $\alpha$ of minimum cost. 

\subsection{{\sf NP} Hardness Proof}

We will prove that {\sc Optimal Assignment} is {\sf NP}-hard by reduction from the $3$-{\sc Partition} problem. Recall that for this problem, the input is a set $X = \{x_1, \ldots x_n\}$ of positive integers, together with a positive integer $T$. The task is to determine whether there exists a partition of $X$ into disjoint subsets $X_1, \ldots, X_{n/3}$ so that for every $j \leq n/3$ it holds that (1) $|X_j| = 3$ and (2) $\sum_{x_i \in X_j} x_i = T$. We refer to such a partition as a {\em solution partition}. It is known (see for example~\cite{HulettWW08}) that $3$-{\sc Partition} remains {\sf NP}-hard even when every $x_i \in X$ satisfies $T/4 < x_i < T/2$. This guarantees that any partition of $X$ into $X_1, \ldots, X_\ell$ that satisfies requirement (2) also automatically satisfies requirement (1) and  therefore is a solution partition. 

\begin{theorem}\label{thm:optAssignmentNPhard} {\sc Optimal Assignment} is {\sf NP}-hard. 
\end{theorem}

\begin{proof}
We give a reduction from the variant of $3$-{\sc Partition} where every $x_i \in X$ satisfies $T/4 < x_i < T/2$. Given an input $(X = \{x_1, \ldots, x_n\}, T)$ to $3$-{\sc Partition}, we create an input to {\sc Optimal Assignment} as follows. The set $R$ has $n$ registers $\{r_1, \ldots, r_n\}$ and for every $i \leq n$, the register $r_i$ has capacity $c_{i}$. The set $O$ has $n/3$ operators $\{o_1, o_2, \ldots, o_{n/3}\}$. Every operator $o_j$ has size $s_j = T$,  satisfied cost $c_j^s = 0$, an unsatisfied cost $c_j^u = 1$. This concludes the construction. We now prove that the instance $X$, $T$ of $3$-{\sc Partition} has a solution partition if and only if the constructed instance $I$, $R$ has an assignment $\alpha$ of cost at most $\sum_{o_j \in O} c_j^s = 0$.

For the forward direction, suppose that there exists a solution partition $X_1, \ldots, X_{n/3}$ of $X$. Define the assignment $\alpha : R \rightarrow O$ as follows: For every $r_i \in R$ we find the unique $j$ so that $r_i \in X_j$ and set $\alpha(r_i) = o_j$.
To argue that $c(\alpha) = \sum_{o_j \in O} c_j^s$, it is sufficient to show that for every operator $o_j \in O$ we have $\rho(o_j, \alpha) = 1$. In particular we have that 
\[
\rho(o_j, \alpha) = \frac{\sum_{r_i \in R ~:~ \alpha(r_i) = o_j} c_i}{s_j} = \frac{\sum_{x_i \in X ~:~ x_i \in X_j} x_i}{T} = 1
\]
We conclude that $c(\alpha) = \sum_{o_j \in O} c_j^s$, as claimed. 

For the reverse direction, suppose there exists an assignment $\alpha$ of cost $c(\alpha) = \sum_{o_j \in O} c_j^s$. Since 
$$c(\alpha) = \sum_{o_j \in O} c(o_j, \alpha) = \sum_{o_j \in O} c_j^s$$
it holds that $c(o_j, \alpha) =  c_j^s$ for every $o_j \in O$. Hence, for every $o_j \in O$ we have $\sum_{r_i \in R ~:~ \alpha(r_i) = o_j} c_i \geq s_j = T$. Furthermore, since 
\begingroup
\allowdisplaybreaks
\begin{align*}
T \cdot \frac{n}{3} & = \sum_{r_i \in R} c_i \\
& = \sum_{o_j \in O} \sum_{\substack{r_i \in R \mbox{ s.t.} \\ \alpha(r_i) = o_j}} c_i \\
& \geq \frac{n}{3} \cdot T
\end{align*}
\endgroup
and the first term of the chain is equal to the last, the last inequality must hold with equality. It follows that $\sum_{r_i \in R ~:~ \alpha(r_i) = o_j} c_i = T$ for every $o_j \in O$.

For every $j \leq n/3$ we set $X_j = \{x_i ~:~ \alpha(r_i) = o_j\}$. We have that for every $j$,
$$\sum_{x_i \in X_j} x_i = \sum_{r_i \in R ~:~ \alpha(r_i) = o_j} c_i = T$$
Thus the partition $X_1, \ldots, X_{n/3}$ of $X$ satisfies requirement (2) of  solution partitions. Since every $x_i$ is in the range $(T/4, T/2)$ the partition also satisfies the requirement (1), that $|X_j| = 3$, of solution partitions. We conclude that $X_1, \ldots, X_{n/3}$ is a solution partition, completing the proof. 
\end{proof}

\subsection{Remarks}
The proof of Theorem~\ref{thm:optAssignmentNPhard} shows that it is in fact {\sf NP}-hard to determine whether a given instance admits an assignment of cost $0$. Thus, assuming ${\sf P} \neq {\sf NP}$, the {\sc Optimal Assignment} problem can not admit an approximation algorithm with {\em any} factor. 

Moreover, the choice of satisfied cost of $0$ and un-satisfied cost of $1$ for every operator in the proof of Theorem~\ref{thm:optAssignmentNPhard} might lead the reader to believe that the {\sf NP}-hardness result only kicks in for instances with unrealistic choices of costs. However it is easily verified that the proof of Theorem~\ref{thm:optAssignmentNPhard} never relies on the precise choice of costs and goes through for any selection of operator costs, as long as $c_j^s < c_j^u$ for every operator $o_j \in O$. 

Finally, Theorem~\ref{thm:optAssignmentNPhard} shows that {\sc Optimal Assignment} is strongly {\sf NP}-hard (\ie that it remains {\sf NP}-hard even when all input integers are coded in unary).

\subsection{Generalized Optimal Assignment}
\label{ssec:greedy_deep}
We now turn to the greedy heuristic that we use to compute the operator-to-register mappings at runtime. However, first we need to formally define the problem for which we design an algorithm, we call the problem {\sc General Optimal Assignment}, {\sc GOA} for short. The {\sc GOA} problem  is a generalization of the {\sc Optimal Assignment} problem defined in Section~\ref{sec:np-hard}. 

Thus, input for {\sc GOA} still contains a set $R$ of registers, a set $O$ of operators, a capacity $c_r$ for every register, and a size $s_o$, a satisfied cost $c_o^s$ and an unsatisfied cost $c_o^u$ for every operator $o$. The task is to compute an assignment $\alpha$ of minimum cost. 

However {\sc GOA} differs from {\sc Optimal Assignment} in three ways. First, we allow the assignment $\alpha$ to be a {\em partial} assignment. That is, instead of being a function $\alpha : R \rightarrow O$, $\alpha$ is a function  $\alpha : R' \rightarrow O$ for a subset $R' of R$. We will say that $R'$ is the {\em domain} of $\alpha$, that registers in $R'$ are {\em assigned} by $\alpha$ and the registers in $R \setminus R'$ are {\em unassigned} by $\alpha$. The definition of cost of an operator with respect to an assignment readily extends without modification to partial assignments. 

The second way in which {\sc GOA} differs from {\sc Optimal Assignment} is that there is additional input that restricts the set of partial assignments $\alpha$ that we can select from. Every register $r \in R$ also comes with a positive integer {\em stage} $t_r$. Finally, input contains a partition of operators $O$ into disjoint non-empty sets $O_1, \ldots, O_\ell$ and a bijection $\pi_i : O_i \rightarrow \{1, \ldots, |O_i|\}$ for every $i \leq \ell$. If $o_1$, and $o_2$ are operators in the same set $O_i$ and $\pi_i(o_2) = \pi_i(o_1) + 1$ then $o_1$ is said to be the {\em parent} of $o_2$, and $o_2$ is the {\em child} of $o_1$. The sets $O_i$ are called {\em dependency chains}.

Finally, the cost of an assignment in {\sc GOA} is different from the sum of costs of all operators with respect to an assignment. The {\em first unsatisfied operator} $u_i$ in the dependency chain $O_i$ with respect to the assignment $\alpha$ is defined as $u_i = \argmin_{o \in O_i ~:~ \rho(o, \alpha) \neq 1} \pi_i(o)$. The cost of a dependency chain $O_i$ with respect to the assignment $\alpha$ is denoted by $c(O_i, \alpha)$ and defined as
$$c(O_i, \alpha) = \begin{cases} c(u_i, \alpha),& \text{if } \exists o \in O_i \text{ s.t. } \rho(o, \alpha) \neq 1\\ c(\pi_i^{-1}(|O_i|), \alpha), & \text{otherwise}\end{cases}$$
The cost of an assignment $\alpha$ in {\sc GOA} is denoted by $c(\alpha)$ and defined to be the sum $c(\alpha) = \sum_{i=1}^l c(O_i, \alpha)$ of the costs of all dependency chains $O_i$ with respect to $\alpha$.

A partial assignment $\alpha : R' \rightarrow O$ is said to be {\em feasible} if the following conditions are satisfied for every pair $o_1$, $o_2$ of operators such that $o_1$ is the parent of $o_2$, and so on. (i) If any register $r$ is mapped by $\alpha$ to $o_2$, then $o_1$ is satisfied by $\alpha$.
(ii) For every pair $r_1$, $r_2$ of registers so that $\alpha(r_1) = o_1$ and $\alpha(r_2) = o_2$ then $t_{r_1} < t_{r_2}$. The intuitive interpretation of these conditions is that if we want to satisfy a child $o_2$ then we first have to satisfy the parent $o_1$, and wait until the next stage before starting to satisfy the child $o_2$.

We summarize the statement of {\sc GOA}: Input consists of a set $R$ of registers, a set $O$ of operators, a capacity $c_r$ and stage $t_r$ for every register, and a size $s_o$, a satisfied cost $c_o^s$ and an unsatisfied cost $c_o^u$ for every operator $o$. Furthermore, input consists a partition $O_1, \ldots O_\ell$ of $O$ into dependency chains and an ordering  $\pi_i : O_i \rightarrow \{1, \ldots, |O_i|\}$ for every dependency chain $O_i$. Note that a dependency chain here represents the sequence of stateful operators for the input queries. For Query~\ref{superspreader}, the set of stateful operators {\tt distinct} $\rightarrow$ {\tt reduce} is an example of a dependency chain.  The task is to compute a feasible (partial) assignment $\alpha$ of minimum cost. 


\subsection{Greedy Heuristic: The Algorithm}\label{ssec:heuralg1}
Our algorithm follows a natural greedy strategy: start with the partial assignment $\alpha_0$ which leaves all registers unassigned. This is a feasible solution. The algorithm then gradually extends the assignment in a series of iterations, maintaining a feasible assignment at every intermediate step. When it is no longer possible to assign even more registers and maintain feasibility the algorithm halts. 

We will say that an assignment $\alpha_2 : R'' \rightarrow O$ is an {\em extension} of an assignment $\alpha_1 : R' \rightarrow O$ if $R' \subset R''$ and $\alpha_1(r) = \alpha_2(r)$ for every $r \in R'$. In other words $\alpha_2$ assigns all the registers that $\alpha_1$ does in precisely the same way, and assigns some additional registers. A {\em feasible extension} of $\alpha$ is an extension $\alpha'$ of $\alpha$ that is itself a feasible assignment. We are now ready to describe the general overview of the algorithm.

\begin{algorithm}[h]
\caption{General framework for computing an assignment}
\SetKwData{Empty}{empty}
$\alpha \gets \Empty$\;
\While{$\alpha$ has at least one feasible extension}{
Compute a set $E$ of feasible extensions of $\alpha$\;
Select an extension $e$ from $E$\;
$\alpha \gets e$\;
}  
\end{algorithm}


Each possible way to compute the candidate set $E$ and selecting $e$ from $E$ leads to a new heuristic. There are infinite possibilities for complex rules here; we have tried to design our rules to be as simple as possible (both conceptually and to implement) while still avoiding some of the pitfalls that lead the algorithm to output assignments with too high a cost. We now discuss precisely what set $E$ of feasible extensions our algorithm computes, and how it selects the extension $e \in E$ to use. We start with the rule for selecting $e$ from $E$.

Our rule for selecting $e \in E$ is perhaps the most natural one - select the $e$ from $E$ that gives the most “bang per buck”, where “bang” refers to the cost of the assignment and “buck” refers to the total capacity of the assignment. More formally, we define the total capacity of an assignment $\alpha$ with domain $R'$ as ${\sf cap}(\alpha) = \sum_{r \in R'} c_r$. Our algorithm selects $e$ from $E$ that maximizes $\frac{c(\alpha) - c(e)}{{\sf cap}(e) - {\sf cap}(\alpha)}$. This kind of selection rule has found tremendous success for other optimization problems, perhaps most famously for {\sc Set Cover}~\cite{setcover}, or maximization of monotone submodular functions~\cite{submodularsetfuncs}. 

Having set the selection criteria we now turn to discussing how we compute the set $E$ of candidate extensions. It turns out that to get the algorithm to produce high-quality solutions we need to consider extensions that assign multiple registers to multiple operators in addition to $\alpha$.
%
%
At the same time we want to avoid having a too large set $E$, since this makes the algorithm slow. The way we handle this trade-off is to go over all possibilities of which operator to satisfy, but to greedily choose the registers which are assigned to the operator.  The next sub-routine, {\bf extend} encapsulates this idea. 

We will say that a dependency chain $O_i$ is {\em satisfied} by $\alpha$ if all operators in $O_i$ are satisfied by $\alpha$. If $O_i$ is not satisfied by $\alpha$ then $O_i$ contains a unique operator $o$ such that $o$ is not satisfied by $\alpha$ and every $o' \in O_i$ with $\pi_i(o') < \pi_i(o)$ is satisfied by $\alpha$. We will call this operator $o$ the {\em active operator} in the dependency chain $O_i$. Our algorithm will always consider a single dependency chain $O_i$ at a time, and will always seek extensions that assign registers to the active operator in $O_i$. When the active operator becomes satisfied, the next operator in the dependency chain becomes the active operator. 

The {\em active stage} of the dependency chain $O_i$ is the lowest value for $t$ so that {\em (i)} there exists at least one register $r$ with stage $t_r = t$ that is unassigned by $\alpha$ and {\em (ii)} if the active operator $o$ of $O_i$ has a parent $p$ then every register $r$ with $\alpha(r) = p$ satisfies $t_r < t$.  The {\em active stage} of the assignment $\alpha$ is the minimum active stage of all dependency chains. 

\begin{algorithm}[t]
\caption{The extend subroutine}
\SetKwData{o}{o}\SetKwData{t}{t}\SetKwData{r}{r}\SetKwData{R}{R}\SetKwData{Best}{best}
\SetKwData{Empty}{empty}\SetKwData{And}{and}\SetKwData{Is}{is}\SetKwData{In}{in}
\SetKwFunction{ActiveOperator}{ActiveOperator}\SetKwFunction{ActiveStage}{ActiveStage}
\SetKwInOut{Input}{input}\SetKwInOut{Output}{output}
\Input{$\alpha$, $O_i$}
\Output{Extension $e$ of $\alpha$}
\o $\gets$ \ActiveOperator{$O_i$, $\alpha$}\;
\t $\gets$ \ActiveStage{$O_i$, $\alpha$}\;
\Best $\gets$ \Empty\;
\ForEach{{\emph unassigned register} \r \In stage \t}{
$e \gets \alpha \cup \{\Best \rightarrow \o\}$\;
\uIf{$\rho(o, e) = 1$ \And $c_r < c_{\Best}$}{
\Best $\gets$ \r\;
}
}
\If{\Best \Is \Empty}{
\Best $\gets$ Largest {\emph unassigned register} \r \In stage \t\;
}
$e \gets \alpha \cup \{\Best \rightarrow \o\}$
\end{algorithm}

The {\bf extend} sub-routine takes as input the assignment $\alpha$ we wish to extend and a dependency chain $O_i$ which is not satisfied by $\alpha$. It then returns a feasible extension $\alpha'$ of $\alpha$ which, in addition to the domain of $\alpha$, assigns a single register $r$ to the active operator $o$ of $O_i$. The sub-routine may also return that no such extension exists. We now describe the {\bf extend} subroutine in detail. 

The {\bf extend} subroutine takes as input $\alpha$ and $O_i$. It then computes the active operator $o$ of $O_i$ and the active stage $t$ of $O_i$. Among all the un-assigned registers $r$ with $t_r = t$ such that $c_r + \sum_{r' \in R ~:~ \alpha(r') = o} c_{r'} \geq s_o$, select the one with smallest $c_r$. Return the extension $e$ of $\alpha$ that additionally assigns $r$ to $o$. If no un-assigned $r$ with $t_r = t$ such that $c_r + \sum_{r' \in R ~:~ \alpha(r') = o} c_{r'} \geq s_o$ exists return the extension $e$ of $\alpha$ that additionally assigns $r'$ to $o$, where $r'$ is the register with largest $c_{r'}$ among all un-assigned  registers with $t_{r'} = t$.

\begin{algorithm}[h]
\caption{Computing set of feasible extensions}
\SetKwData{t}{t}\SetKwData{Empty}{empty}\SetKwData{Is}{is}\SetKwData{to}{t'}\SetKwData{E}{$E$}\SetKwData{In}{in}
\SetKwFunction{ActiveStage}{ActiveStage}\SetKwFunction{Extend}{extend}
\SetKwInOut{Input}{input}\SetKwInOut{Output}{output}
\Input{A partial assignment $\alpha$}
\Output{Set of feasible extensions \E of $\alpha$}
\t $\gets$ \ActiveStage{$\alpha$}\;
\E $\gets$ \Empty\;
\ForEach{{\emph unsatisfied dependency chain} $O_i$ \In $\alpha$}{
\to $\gets$ \ActiveStage{$O_i$, $\alpha$}\;
\If{$\to = \t$}{
$e_i^0 \gets \alpha$\;
$j \gets 0$\;
\While{$O_i$ \Is \textit{unsatisfied}}{
$e_i^{j+1} \gets$ \Extend{$e_i^j$, $O_i$}\;
$\E \gets \E \cup e_i^{j+1}$\;
$j \gets j + 1$\;
}
}
}
\end{algorithm}
We are now ready to describe how the set $E$ is computed. The algorithm first computes the active stage $t$ of the assignment $\alpha$. It then iterates over all dependency chains $O_i$ which are not satisfied by $\alpha$, and whose active stage is precisely equal to $t$. For each such dependency chain $O_i$ the algorithm proceeds as follows. Initially it sets $e_i^0 := \alpha$ and $j = 0$. Then, as long as $O_i$ is not satisfied by $e_i^j$ it adds $e_i^j$ to $E$, sets $e_i^{j+1}$ to be the output of ${\bf extend}$ on $e_i^j$ and $O_i$, and increments $j$. The resulting set $E$ is union of the output of this algorithm for all the choices of dependency chain $O_i$.



\subsection{Remarks on Heuristic for {\sc GOA}}
We remark that our actual implementation of the heuristic of Section~\ref{ssec:heuralg1} contains several optimizations. For example the actual implementation does not first compute the set $E$ and then select the extension $e$ from $E$. Instead it enumerates all of the elements $e \in E$ in the manner described in Section~\ref{ssec:heuralg1}, computes their "bang-per-buck" score on the fly, and keeps the $e \in E$ with highest score. 

By enumerating the elements in $E$ in a carefully chosen order and using appropriate data structures we can re-use a lot of computations when calculating the score of each element $e$. We omit the details of these optimizations, since they do not change which assignment the algorithm produces. 

We also consider an enhanced version of the heuristic of Section~\ref{ssec:heuralg1}. Here, whenever the active stage $t$ of the current assignment $\alpha$ increments, we remove from the domain of $\alpha$ all registers $r$ whose stage $t_r$ is at least $t$. This allows the greedy algorithm to "undo" choices for stage $t$ which were made when the active stage was strictly less than $t$ and the algorithm had less information about which operators still need registers assigned to them. \system uses this enhanced version of the greedy heuristic to compute operator mappings.

%% file: tofino.tex
\section{Tofino-based Implementation}
\label{sec:tofino}
In this part of the appendix, we detail how we implemented different functionalities in a Tofino-based switch to enable dynamic operator mappings. We also present a case study that demonstrates how this implementation handles input workload dynamics, a capability that is missing from Sonata. Finally, we quantify the runtime overheads for this implementations.

\subsection{Dynamic Operator Mappings}
As described in Section~\ref{sec:implement}, enabling dynamic operator mappings in the data plane entails dynamically changing: (1)~the program for each stateful arithmetic logic unit (SALU) for each register in the data plane, (2)~the set of keys used for stateful operations, and (3)~the packet-processing pipeline in the data plane. 

\smartparagraph{Updating SALU programs.}
Figure~\ref{fig:salu-mat} shows how \system uses a table, 
{\tt execute\_op}, to dynamically change a SALU’s program. For this table, it defines actions executing the {\tt reduce} and the {\tt distinct} program, respectively. This table reads a packet's {\tt select\_prog} metadata field to decide which action (\ie SALU program) to select for stateful operations. 

\begin{figure}[t!]
\begin{lstlisting}[language=P4,basicstyle=\footnotesize, 
basicstyle=\footnotesize, numbers=left,xleftmargin=2em,frame=single,
framexleftmargin=2.0em, captionpos=b.]
action do_execute_reduce() {
    reduce_program.execute_stateful_alu(meta_op.index);
}

action do_execute_distinct() {
    distinct_program.execute_stateful_alu(meta_op.index);
}

action drop_exec_op(){
    modify_field(
        meta_app_data.drop_exec_op, 1);
 }

table execute_op {
    reads {
        meta_op.select_prog : exact;
    }
    actions {
        do_execute_reduce;
        do_execute_distinct;
        drop_exec_op;
    }
    size : 2;
}

\end{lstlisting}
\caption{Using packet's metadata field and match-action table to select the SALU program for stateful operations.\label{fig:salu-mat}}
\end{figure}




\smartparagraph{Updating the set of keys for stateful operations.}
To understand how \system dynamically updates the keys for stateful operations, consider the case where the two keys, say, {\tt dIP} and {\tt sIP}, and \system needs to dynamically decide which of the two to use for computing the index value for the stateful operation. Figure~\ref{fig:key-mat} shows how \system uses the {\tt init\_hash\_field\_data} table to apply two different actions. Each of these actions read the packet field value from the packet, apply the mask (for iterative refinement), and write the masked value to the packet's metadata field. Figure~\ref{fig:key-hash} shows how \system uses this metadata field to compute the hash index for stateful operation. Extending this code block to enable hashing on multiple packet fields is straightforward. Here, instead of creating a single metadata field, \system creates multiple metadata fields and then changes their values using tables specific to each of these fields. 

\begin{figure}[t!]
\begin{lstlisting}[language=P4,basicstyle=\footnotesize, 
basicstyle=\footnotesize, numbers=left,xleftmargin=2em,frame=single,
framexleftmargin=2.0em, captionpos=b.]
action do_init_src_ip(dynamic_mask) {
    bit_and(meta_op.field_value, ipv4.srcIP, dynamic_mask);
}

action do_init_dst_ip(dynamic_mask) {
    bit_and(meta_op.field_value, ipv4.dstIP, dynamic_mask);
}

table init_hash_field_data {
    actions {
        do_init_src_ip;
        do_init_dst_ip;
    }
    size : 1;
}
\end{lstlisting}
\caption{Match-action tables to select the key for stateful operation.\label{fig:key-mat}}
\end{figure}

\smartparagraph{Updating the packet-processing pipeline}
Dynamically mapping stateful operators to different registers in the data plane requires updating the packet processing pipeline to ensure sequential composition of these operators. Sonata used a combination of the query-specific metadata field {\tt qid\_drop} and {\tt IF} statements to compose dataflow operators in the data plane. Such an approach is only suited for static operator-to-register mappings. \system's control program applies all the match-action tables in sequence, and uses a packet's metadata-fields to decide whether or not to apply particular actions in the match-action table to it. Such an approach enables \system to dynamically reconfigure the packet-processing pipeline at runtime.



% ...


To enable such flexible packet processing, \system narrows the scope of the {\tt drop} metadata field to match-action tables for stateful operators. The decision of whether the remaining downstream operators should be applied to this packet is encoded in this {\tt drop} field. Such a change ensures that \system can compose these operators without requiring {\tt IF} statements.

\system also enables support for multiple compile time dependencies. For example, if there are two operators $a$ and $b$ in the first stage and $c$ and $d$ in the second stage, it enumerates all possible dependencies between them, \ie $a \rightarrow c$, $a \rightarrow d$, $b \rightarrow c$ and $b \rightarrow d$. To realize these dependencies, it makes the operators in second stage (\ie $c$ and $d$) match on both $a$ and $b$'s `drop` metadata. At runtime, \system adds match-action table entries such that only one of the two possible dependencies actually occur for these operators. 

Note that although making an exact match on $N$ drop fields (where $N$ denotes the number of operators in the last stage) requires $2^N$ match-action table entries at runtime, an operator in the second stage only follows on a single operator in the first stage, which in turn enables \system to use a single ternary match-action rule. 

\begin{figure}[t!]
\begin{lstlisting}[language=P4,basicstyle=\footnotesize, 
basicstyle=\footnotesize, numbers=left,xleftmargin=2em,frame=single,
framexleftmargin=2.0em, captionpos=b.]
field_list hash_op_fields {
    meta_op.field_value;
}

field_list_calculation hash_op_calc {
    input {
        hash_op_fields;
    }
    algorithm: hash_algo;
    output_width: hash_width;
}

action do_init_hash() {
    modify_field_with_hash_based_offset(
        meta_op.index, 0,
        hash_op_calc, reg_1_size);
}

table init_hash {
    actions {
        do_init_hash;
    }
    default_action : do_init_hash;
    size : 1;
}

\end{lstlisting}
\caption{Match-action tables to compute the hash value for the selected the key for stateful operation.\label{fig:key-hash}}
\end{figure}
\smartparagraph{Metadata fields.}
So far, we described how \system used different metadata fields to dynamically update the packet processing pipeline at runtime. We now summarize the type of metadata fields \system uses and provide an estimation of their memory overhead. 

\begin{asparaitem}

\item \textbf{Ephemeral fields.} For each data-plane register, \system needs the following metadata fields: (1)~a one bit {\tt select\_program} field to change the SALU program at runtime; (2)~a one bit {\tt drop} field to mark if the packet needs further processing; (3)~a {\tt index} field of size $log_{2}$(register size) bits to store the index for executing a stateful operation; (4)~a {\tt keys} field for storing the packet fields it uses as keys for the stateful operation. All these fields are narrowly scoped, can be reused across stages, and are therefore termed \textit{ephemeral}. A single register requires around 114 bits (1 for SALU program + 1 for drop + 96 for keys + 16 for index). Thus, for the target considered in this paper (8 stages), \system requires around 1.8~Kb for a packet's metadata for ephemeral fields. 
\item \textbf{Persistent fields.} For each query, \system needs the following metadata fields: (1)~a $16$-bit {\tt qid} field for each query that gets executed in the data plane; (2)~{\tt keys} to store the packet fields that are used as keys for computing the index value for query's last stateful operator in the data plane; and (3)~{\tt index} field to store the hash value for the key fields. The index field enables \system's runtime to read the aggregated {\tt value} from the data plane, which in turn enables \system to emit the output (key, value) pair to the stream processor for further processing. For each query, \system requires 128 bits in a packet's metadata for these persistent fields, which is similar to Sonata's metadata footprint. This number will increase as the number of input queries increases. 
\end{asparaitem}

\system introduces an additional optimization to avoid wasting register memory for {\tt distinct} operation. More concretely, for {\tt reduce} operation, \system configures a two-dimensional register memory specifying the number of rows, where each row is of size 32 bits. For the match-action table applying the {\tt reduce} operator, it updates the value in the row identified by the index field. Naively, using the same approach for {\tt distinct} operator will waste 31 bits for each row as it only requires one bit to store the state. To ensure that the SALU programs for {\tt distinct} operator can make use of all available bits in the register, \system uses five additional bits for computing the hash. These five bits enable the match-action table to select a bit among the 32 bits in each row. This optimization further requires additional 80 bits in packet's metadata for the target considered in the paper.

\subsection{Case Study}
\label{ssec:tofino-case-study}
We now demonstrate how the different building blocks described above enable \system to dynamically update the packet processing pipeline in the data plane, enabling it to update operator-to-register mappings at runtime.

\smartparagraph{Setup.}
For this case study, we use the two queries Query~\ref{newtcp} and Query~\ref{ddos} from Table~\ref{tab:queries}. To simplify the experiment, we do not consider iterative refinement (\ie both queries use all 32 bits of the {\tt dIP} and {\tt sIP} fields for stateful operations). Also, we only execute one stateful operator for each query in the data plane, \ie we execute the {\tt reduce} for Query~\ref{ddos} in userspace. 

We use the Tofino-model~\cite{p4studio} on a \texttt{n1-highcpu-8} machine over Google Cloud Platform (GCP)~\cite{gcp}. We use Scapy~\cite{scapy} to send the traffic to this switch. The input workload is bi-modal, \ie at first, the {\tt reduce} operator for Query~\ref{newtcp} receives more keys than the {\tt distinct} operator for Query~\ref{ddos}, and then the opposite scenario occurs. More specifically, for the first three windows, Query~\ref{newtcp} receives traffic with around 10 distinct {\tt dIP}s, and Query~\ref{ddos} receives traffic with around 100 distinct ({\tt sIP,dIP}) pairs per window. After three second, these numbers change to around 100 and 10, respectively. We run this experiment with both \system and Sonata, each using two registers (sized 64$\times$32 and 2048$\times$32 bits) in the data-plane to execute these queries. When the change in traffic patterns occurs, \system swaps the mapping between stateful operators and data plane registers while Sonata continues to apply the same operator to register mappings. 

\begin{figure}[t] 
\begin{minipage}{1\linewidth}
\includegraphics[width=\linewidth]{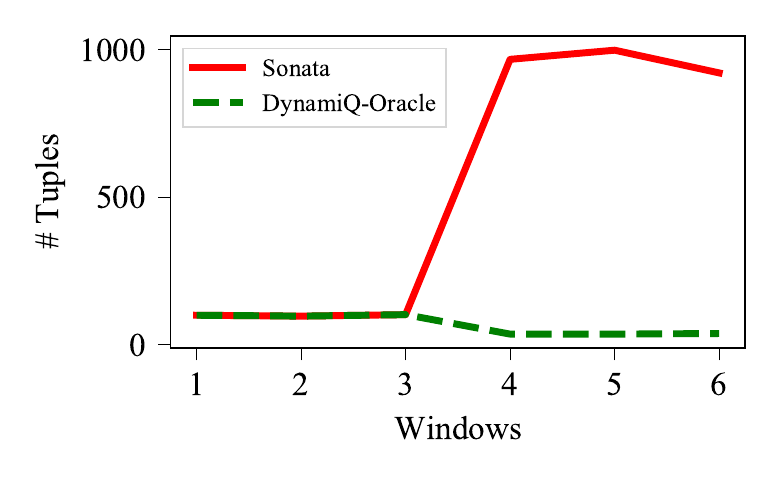}
\end{minipage}
\caption{Case study with Tofino-model. Here, incoming traffic workload changes after three seconds. As a result, the workload at stream processor increases for Sonata, and remains almost the same for \system.}
\label{fig:tofino-study}
\end{figure}

\smartparagraph{Observations.}
Figure~\ref{fig:tofino-study} shows the number of tuples received by the stream processor. We observe that the workload at the stream processor is the same for both systems for first three one-second windows. However, when the traffic pattern changes at $t=3$, we observe that while the workload for Sonata increases, it remains essentially unchanged for \system. Note that the slight reduction in the workload at the stream processor for \system is because the {\tt reduce} operator for Query~\ref{newtcp} that receives more keys after the traffic pattern flips only reports the ones that exceed the threshold. Here, to handle the workload dynamics, \system is changing the operator-to-register mappings, which entails: (1)~ changing the SALU program from {\tt distinct} to {\tt reduce}, and vice versa; (2)~changing the index key from {\tt sIP} to {\tt dIP}, and vice versa; and (3)~updating the packet processing pipeline to ensure that the stream processor can correctly apply the remaining dataflow operations in the userspace. This experiment demonstrates \system's ability to dynamically update operator-to-register mappings at runtime to handle workload dynamics.  

We provide the instructions to reproduce this case study with a Tofino-model switch over GitHub~\cite{tofino-github}.




\subsection{Runtime Overheads}
\system requires updating the match-action tables in the data plane at runtime. The number of match-action table entries that \system updates are proportional to the number of ALUs or registers in the data plane. In the worst case, for each data-plane register, \system needs to update eight match-action table entries. Here, two are required to configure the tables that enable selecting the right SALU program, five to choose the set of fields for computing the index, and one for updating the {\tt drop} field. Thus, for the data-plane target considered in the paper, \system will require updating 576 match-action table entries. Our current implementation with Tofino's PD-API~\cite{p4studio}, which updates the match-action table entries one at a time, will take around 200-300~ms to update these many match-action table entries, which is relatively inefficient. However, as discussed in Section 5, this is not a fundamental limitation for \system. A more efficient API that can update the match-action table rules in parallel can significantly reduce this overhead. 


\begin{figure}[t!]
\begin{lstlisting}[language=query2,basicstyle=\footnotesize, 
basicstyle=\footnotesize, numbers=left,xleftmargin=2em,frame=single,
framexleftmargin=2.0em, captionpos=b, label=newtcp, caption=Newly opened TCP connections.]
packetStream(W)
.filter(p => p.tcp.flags == 2)
.map(p => (p.dIP, 1))
.reduce(keys=(dIP,), f=sum)
.filter((dIP, count) => count > Th)
\end{lstlisting}
\ifx \compress \undefined
\else
\vspace*{-1\baselineskip}
\fi
\end{figure}